\documentclass[12pt]{article}

\usepackage[latin1]{inputenc}
\usepackage{amsmath}
\usepackage{amsfonts}
\usepackage{amssymb}
\usepackage{amsthm}
\usepackage{bm}
\usepackage{hyperref}
\usepackage{graphicx}
\usepackage{natbib}
\usepackage[margin=1.2in]{geometry}
\usepackage{color}
\usepackage{setspace}
\usepackage{tabularx}
\usepackage{booktabs}

\newcommand{\Real}{\mathbb{R}}

\newcommand{\argmin}{\qopname\relax m{arg\,min}}

\newtheorem{theorem}{Theorem}[section]
\newtheorem{lemma}{Lemma}

\newtheorem{corollary}{Corollary}
\newtheorem{condition}{Condition}

\usepackage{setspace}
\usepackage{xr}  
\begin{document}

\begin{titlepage}
	\title{Quantile Predictions for Equity Premium using Penalized Quantile Regression with Consistent Variable Selection across Multiple Quantiles}
	
	\author{Shaobo Li and Ben Sherwood\thanks{Shaobo Li (\href{mailto:shaobo.li@ku.edu}{shaobo.li@ku.edu}) is Associate Professor, School of Business, University of Kansas. Ben Sherwood (\href{mailto:ben.sherwood@ku.edu}{ben.sherwood@ku.edu}) is Jack and Shirley Howard Mid-Career Professor, Associate Professor, School of Business, University of Kansas}}
	\date{}
	
	\maketitle
	
	\begin{abstract}
		This paper considers equity premium prediction, for which mean regression can be problematic due to heteroscedasticity and heavy-tails of the error. We show advantages of quantile predictions using a novel penalized quantile regression that offers a model for a full spectrum analysis on the equity premium distribution.  To enhance model interpretability and address the well-known issue of crossing quantile predictions in quantile regression, we propose a model that enforces the selection of a common set of variables across all quantiles. Such a selection consistency is achieved by simultaneously estimating all quantiles with a group penalty that ensures sparsity pattern is the same for all quantiles.
		Consistency results are provided that allow the number of predictors to increase with the sample size. A Huberized quantile loss function and an augmented data approach are implemented for computational efficiency. Simulation studies show the effectiveness of the proposed approach. Empirical results show that the proposed method outperforms several benchmark methods.
		Moreover, we find some important predictors reverse their relationship to the excess return from lower to upper quantiles, potentially offering interesting insights to the domain experts. Our proposed method can be applied to other fields.
	\end{abstract}

	\noindent%
	\vfill

\end{titlepage}

\section{Introduction}\label{sec:Intro}

Equity premium predictability has been a long debate among empirical finance researchers and practitioners. Early studies have found dividend ratios such as dividend-price ratio, dividend-payout ratio, and dividend yield are predictive to the future excess return at longer horizons. \citep{campbell1988dividend,fama1988dividend}. \citet{welch2008comprehensive} comprehensively examines a set of well established predictors, including dividend ratios, stock volatility, various interest rates and spreads among others, for predicting the excess return, and they assess both in-sample and out-of-sample (OOS) predictability. Their empirical results show that these predictors do not generate stable out-of-sample prediction using the OLS model, and cannot outperform the simple prediction based on the historic average. In contrast, some other studies have shown empirical evidences that the equity premium is predictable. 
\citet{campbell2008predicting} imposes sign restrictions and show that the modified model outperforms the historic mean. \citet{rapach2010out} demonstrate that combining (averaging) the predictions from individual predictors (e.g., simple linear regression) generates stable OOS prediction. More recently, \citet{goyal2024comprehensive} provides another comprehensive review of this stream of research by examining over 40 variables discussed by numerous articles from top finance journals. Again, their findings are mixed with most predictors fail to deliver significant out-of-sample predictability.

The primary statistical tool used in most existing studies is the ordinary least square (OLS) model. However, the standard approach for estimating quantiles using OLS relies on the key assumption of normally distributed errors with constant variance, assumptions that are often violated in financial data. Stock returns, in particular, tend to be noisy and heteroscedastic, exhibiting non-Gaussian distributions. Within the linear model framework, quantile regression \citep{qrBook}, by contrast, does not rely on these assumptions and is well known for its robust estimation and ability to effectively handle heteroscedasticity. Quantile regression has become an extremely useful statistical model in economics and finance. For instance, in financial risk management, quantile regression is widely used for estimating conditional value-at-risk (VaR) \citep{chernozhukov2001conditional,james2023forecasting}, a key risk measure of potential capital loss at extreme conditions. 
Few studies have explored the use of quantile regression in return prediction. \citet{meligkotsidou2014quantile} argues that modeling various quantiles is helpful to capture the distribution of the excess return as opposed to the mean regression, and they apply traditional quantile regression to examine the excess return predictability. We largely assent their argument of using quantile regression approach for equity premium prediction. 
\citet{lee2016predictive} develops a quantile regression with IVX filter that handles persistent covariates and use the method to test predictivity of eight commonly used predictors for excess return across various quantiles. However, they do not assess out-of-sample prediction accuracy in the empirical analysis, thus the predictive performance is unknown. 
Different from these works, we use more recent data, and develop a model that simultaneously estimates multiple quantiles of excess return conditional on a common set of predictors selected from a large set of variables.

We start by presenting evidence that the OLS model may fail in equity premium prediction while quantile regressions can be a better alternative.	
Figure \ref{fig:motivation1} presents out-of-sample 80\% prediction interval of the excess return, via an expanding window approach detailed in Section \ref{sec:empirical}, comparing OLS and the proposed quantile regression models, which will be introduced subsequently. Clearly the 80\% prediction intervals are largely disagreed between the two methods, where the prediction interval from OLS is constantly wider than that from the quantile regression (shaded area). In fact the latter indeed covers approximately 80\% of the observed excess returns across the evaluation period while the 80\% OLS prediction interval covers about 90\% of the observed data. This difference is potentially due to assumption violations such as the constant variance. As shown in Figure \ref{fig:motivation2}, some of the commonly used predictors present a clearly ``fan-shaped" relationship with the excess return, a strong evidence of a heteroscedasticity, which can be correctly modeled by the quantile regression. Later in our empirical analysis, our method reveals and quantifies such a heteroscedastic relationship with reversed coefficient signs from lower to upper quantiles, which otherwise remains undetectable through traditional mean regression. 
We strongly advocate for the use of quantile regression, not only to address the issues that arise from mean regression but also because it offers a full spectrum of estimation and prediction, which is particularly useful when modeling the stock market. 
\begin{figure}[h!]
	\centering
	\includegraphics[scale=0.65]{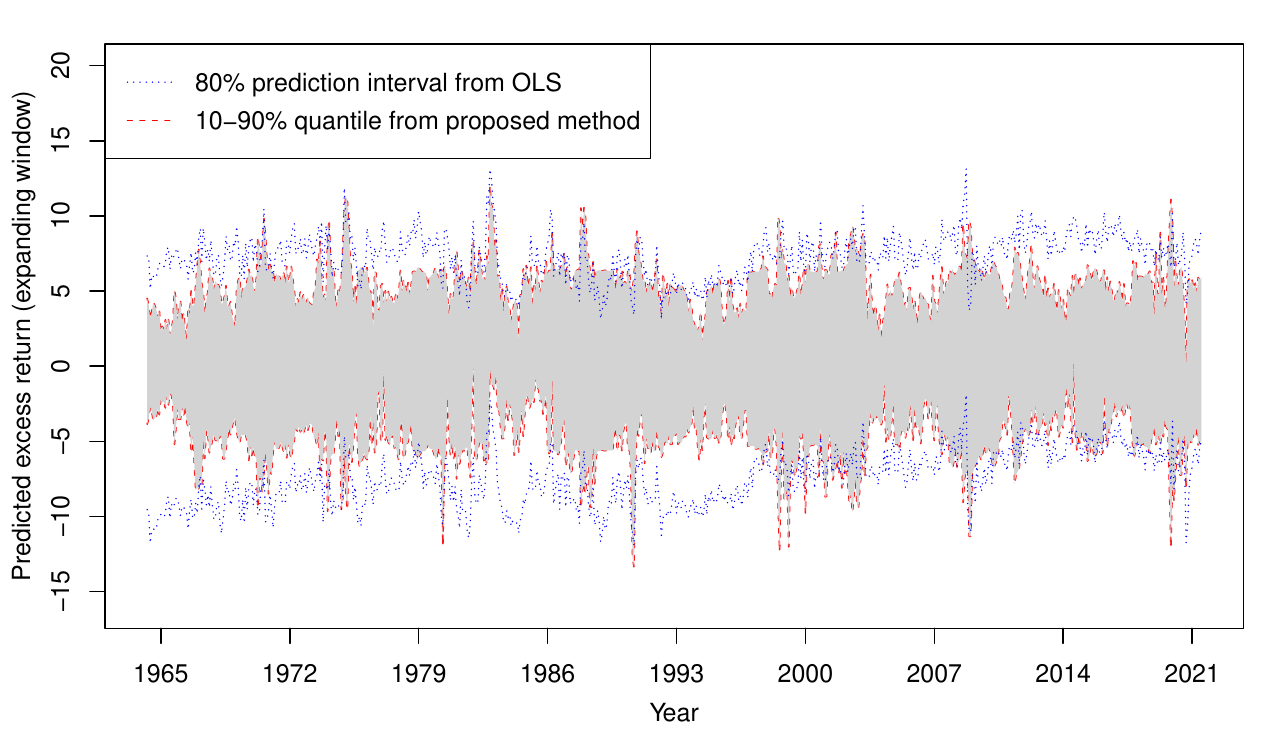}
	\caption{80\% prediction interval based on OLS and the proposed quantile regression.} \label{fig:motivation1}
\end{figure}

\begin{figure}[h!]
	\centering
	\includegraphics[scale=0.6]{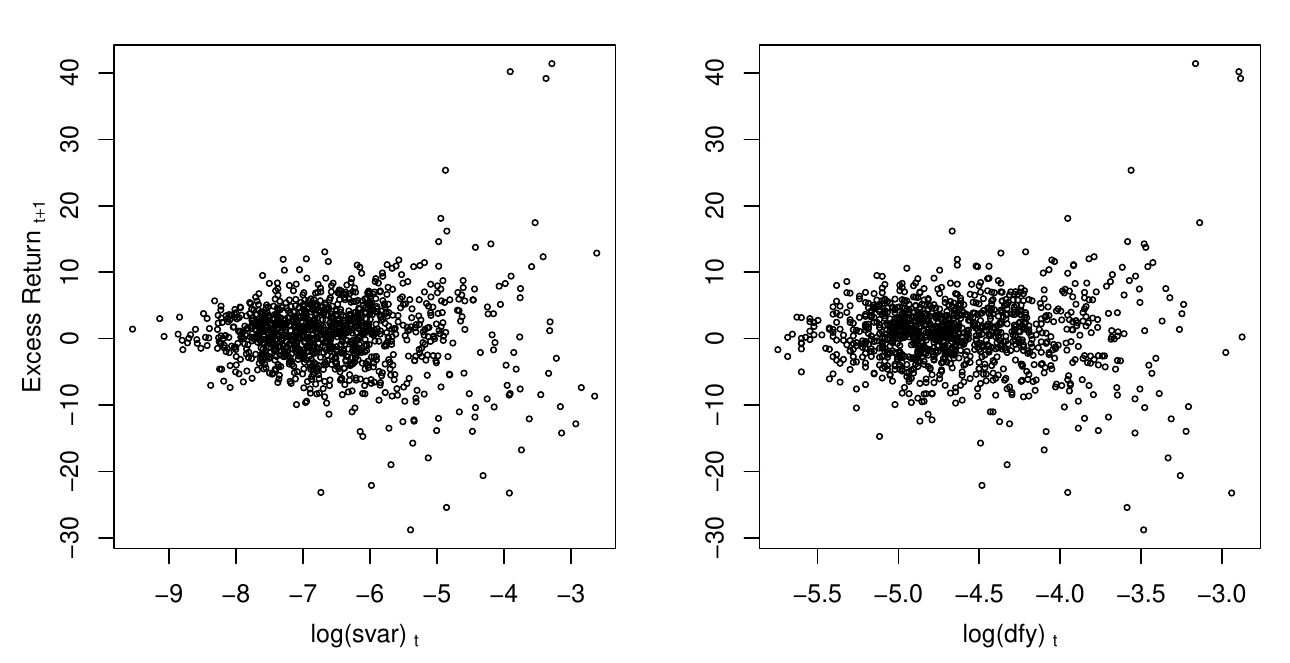}
	\caption{Scatter plot between excess return and two lagged predictors: stock variance (svar) and default yield spread (dfy) with logarithm transformation.} \label{fig:motivation2}
\end{figure}

There is a large literature on determining which variables are useful for equity premium prediction, but most of these works examine predictors individually with underlying asset pricing theories \citep{goyal2024comprehensive}. Modern statistical tools, on the other hand, allow efficient data-driven variable selection, among which penalized regression is one of the most popular ones due to its computational efficiency and simultaneous variable selection and model estimation.  \citet{lee2022lasso} considers LASSO-type \citep{lasso,adaptiveLasso} penalized mean regression for excess return prediction with persistent covariates. 
Much of the research in quantile regression has focused on extending popular penalties from mean regression to quantile regression \citep{qr_lasso,scad_qr,yi2017semismooth}.	These approaches consider estimating each quantile separately, by adding the penalty to the objective function proposed by \citet{origQR}. As a result the selected predictors can vary with the quantile being modeled. When estimating different quantiles for equity premium we find the selected variables often varies across the estimated models at different quantiles, which poses unnecessary complexity for model interpretation. Even worse, it can lead to an incoherent model that suffers from the so-called crossing quantiles, e.g., the predicted values at lower quantiles are larger than those at higher quantiles. In our empirical analysis, we found that using quantile regression with the Lasso penalty results in over 30\% predictions with crossing quantiles for out-of-sample predictions, see Figure \ref{fig:cq} for an illustration. 
Motivated by interest in providing a coherent model of excess return for multiple quantiles, we develop a novel penalized quantile regression method that guarantees the set of selected variables will be the same for all quantiles, significantly mitigating the issue of crossing quantiles. 
Our proposed method offers an easy-to-interpret model for equity premium predictions at multiple quantiles, providing a full spectrum of the estimation, potentially aiding decision-making for both risk-seeking and risk-averse investment preferences.




Quantile regression will often be beneficial in settings where the variance is non-constant. Consider the following location-scale model, which is frequently used to motivate quantile regression, 
\begin{equation}
	\label{locScale}
	y = \bm x^\top \bm \beta^* + (\bm x^\top \bm \zeta^*)\epsilon,
\end{equation}
where $\epsilon \sim F$, $\bm x \in \Real^{p+1}$ is the covariate vector, $\bm \beta^* \in \Real^{p+1}$ is the location coefficient vector including intercept, and $\bm \zeta^*\in \Real^{p+1}$ is the scale coefficient vector. Let $Q_\tau(y|\bm x)$ be the $\tau$th conditional quantile of $y$ given $\bm x$ and $F^{-1}_\epsilon(\tau)$ as the $\tau$th quantile of $\epsilon$. Under the assumption that $\bm x^\top \bm \zeta^*$ is positive then 
\begin{equation}
	Q_\tau(y \mid \bm x) = \bm x^\top[\bm \beta^*+\bm \zeta^* F^{-1}_\epsilon(\tau)] \equiv \bm x^\top \bm \beta^*_\tau,
\end{equation}
where $\bm \beta^*_\tau = (\beta^*_{\tau,0},\ldots,\beta^*_{\tau,p})^\top$. Then for the $j$th predictor, $\beta^*_{\tau,j}=\beta^*_j+\zeta^*_jF_\tau(\epsilon)^{-1}$. If $\beta^*_j=\zeta^*_j=0$ then $\beta^*_{\tau,j}=0$. Otherwise, $\beta^*_{\tau,j}=0$ only when $F^{-1}_\epsilon(\tau)=-\beta^*_j/\zeta^*_j$, which if $F$ is continuous will happen at most only for one value of $\tau$. In other words in \eqref{locScale} the coefficient of a predictor is either zero across all quantiles or non-zero across all quantiles, except for at most one quantile. Given that 
\eqref{locScale} is a popular motivating model for quantile regression, we believe this is strong motivation for considering consistent selection across quantiles.

To provide consistent selection across quantiles, we adopt a group lasso penalty \citep{yuan2007}, where each group is a predictor across all quantiles being modeled. As a result, a predictor is either selected or not selected for all the specified quantiles. We name the proposed method grouped-Quantile-Lasso (gQ-Lasso). A similar approach has been considered to select the same predictors for a multiple response mean regression model \citep{li2015}. We are not the first to consider a penalized approach that provides consistent selection of predictors across quantiles. \citet{ZOU20085296} discussed using the $L^2$ norm but instead proposed using the infinity norm. They specifically emphasize the computational advantages of using the infinity norm because that can be framed as a linear programming problem while incorporating the $L^2$ norm requires a second-order cone programming approach. However, in our simulations we demonstrate that our proposed algorithm for using the $L^2$ norm results in a much faster algorithm then using the infinity norm with a linear programming algorithm. Similar to this work, \citet{parkHeZhouMultQuant} motivate their estimator because of concerns about inconsistent selection across quantiles and crossing quantiles. They propose a method that allows the set of selected predictors to vary slowly across quantiles, but does not guarantee consistent selection across quantiles. Other methods provide a group penalty for the variables across the quantiles that promote consistent selection across quantiles, but do not guarantee it as they also allow the sparsity of the predictors to change across quantiles \citep{Peng2014-ql,HE2016222}. Assuming the location scale model of \eqref{locScale} holds, \citet{heteroIdQR} propose a method that guarantees consistent selection across quantiles and identifies which variables have a heteroscedastic relationship. They propose an approach similar to adaptive lasso \citep{adaptiveLasso} and assume $p < n$, while we consider a standard group lasso penalty and allow for $p > n$.



In addition to method development, algorithm, and theory, our study also contributes to the finance literature, particularly the field of equity premium prediction, in the following aspects. First, we demonstrate that the widely used OLS models can be problematic as the data show clear signs of heteroscedasticity and possible heavy-tailed errors. 	
Second, in terms of prediction, our empirical results show that the proposed method significantly outperforms mean regressions based on either least square or Huber loss for robust estimation, especially for the out-of-sample analysis. 
Finally, our empirical findings offer several insights that may shed new light for domain experts. First, we find little out-of-sample predictability for the full evaluation period (1965-2021) that includes several recessions, while the predictability is modest for periods without major recessions. 
Second, in evaluating prediction quantiles, we find that the predictions for higher quantiles are generally the best compared to middle and lower quantiles. Third and most interestingly, we are able to identify reversed relationships between excess return and several predictors from lower to upper quantiles. For instance, we find that \textit{stock variance (svar)} is negatively associated with excess return at lower quantiles but positive at higher quantiles, while existing studies only show a positive relationship \citep{guo2006out}. 


The rest of the paper is organized as following. Section \ref{sec:model} formally introduces the proposed method that is followed by theoretical properties discussed in Section \ref{sec:theory}. An efficient algorithm is outline in Section \ref{sec:algorithm}. We conduct extensive simulation studies in Section \ref{sec:simulation} and in-depth empirical analyses on equity premium prediction in Section \ref{sec:empirical}. 

\section{Model} \label{sec:model}

\subsection{Notation}
Vectors, such as $\bm a \in \Real^p$, will be represented by bold lowercase letters, while matrices, such as $\bm A \in \Real^{n \times p}$, will be represented by bold capital letters. For any vector $||\bm a||_q$ represents the $L^q$-norm and $||\bm a||_\infty$ is the maximum element of the vector $\bm a$. For any matrix $\bm A$, $||\bm A||$ is the spectral norm. For subspace $\mathcal{A} \subseteq \Real^{p}$, the vector $\bm a_{\mathcal{A}}$ is the projection of $\bm a$ to $\mathcal{A}$ and $\mathcal{A}^\perp$ is the orthogonal complement of $\mathcal{A}$. 

\subsection{The Model -- Grouped-Quantile Lasso}
Consider independent identically distributed random variables $\{y_i,\bm x_i\}_{i=1}^n$ with $y_i \in \Real$ and $\bm x_i = (1,x_{i1},\ldots,x_{ip})^\top \in \Real^{p+1}$. For K quantiles of $\tau_1 < \ldots < \tau_K$ consider the quantile regression model
\begin{equation}
	\label{model}
	y_i = \bm x_i^\top\bm \beta_{k}^* + \epsilon_i^k,
\end{equation}
where $\bm \beta_{k}^*=[\beta_{k0}^*,\ldots,\beta_{kp}^*]^\top \in \Real^{p+1}$ and $P(\epsilon_i^k < 0 | \bm x_i) = \tau_k$. In this paper we propose a method for simultaneously estimating $K$ quantiles of $\{\tau_1,\ldots,\tau_K\}$. Define $\bm \beta^{j*} = [\beta_{1j}^*,\ldots,\beta_{Kj}^*] \in \Real^{K}$ as the vector of coefficients for predictor $j$ for all quantiles.

Traditional approaches may select different variables at different quantiles. This is often not desired in practice. To this end, in this work we propose a new approach that consistently select variables across different quantiles. Specifically, we use the group-Lasso penalty so that for each predictor the $K$ quantiles to be modeled can be regarded as a group. That is, for any $j \in \{1,\ldots,p\}$ and $k' \in \{1,\ldots,k\}$ if $\hat{\beta}_{k'j} = 0$ then $\hat{\beta}_{kj} = 0$ for all $k \in \{1,\ldots,K\}$. To achieve this we propose minimizing the following penalized objective function, 
\begin{align}
	L_{\lambda}(\bm \beta| \tau_1,...,\tau_K; \mathbf{X}, \mathbf{Y})=\frac{1}{n}\sum_{k=1}^{K}\sum_{i=1}^{n}\rho_{\tau_k}(y_i- \bm x_i^\top \bm \beta_k)+\lambda \sum_{j=1}^{p} \|\bm \beta^j\|_2, \label{eq:obj1}
\end{align}
where $\rho_{\tau_k}()$ is the quantile loss function given quantile $\tau_k$, $\bm \beta=(\bm \beta_1^\top, ..., \bm \beta_K^\top)^\top \in \mathbb{R}^{K(p+1)}$ and $\bm \beta^j \in \mathbb{R}^K$ are the coefficients of predictor $X_j$ for $K$ different quantiles.   
The group lasso penalty in \eqref{eq:obj1} guarantees 
that a variable is either selected or not for all quantiles being modeled. \citet{li2015} proposed a similar method for estimating multiple mean regression models, where the response variables are different. Here, we are estimating multiple conditional quantiles of the same response. 


\section{Theoretical Properties} \label{sec:theory}
Define $\bm \beta^* = [\bm \beta_{1}^{*\top},\ldots,\bm \beta^{*\top}_{K}]^\top \in \Real^{K(p+1)}$ as the vector of all quantile coefficients. Let $\mathcal{S} = \{ j \in \{1,\ldots,p\} |\, ||\bm \beta^{j*}||_2 \neq 0\}$, be the set of predictors that are active at some quantile and $s=|\mathcal{S}|$ be the cardinality of that set. Let $\bar{\mathcal{S}} = \mathcal{S} \cup \{0\}$ be the intercept and set of active predictors. Note, as outlined in the introduction, it is possible that $\beta_{kj}^* \neq 0$ and $\beta_{k'j}^* = 0$ for some $k,k' \in \{1,\ldots,K\}$. However, we take the stance that it is desirable to select variables across all quantiles for the related issues of interpret-ability, and it being hard to identify a model where there tends to be disagreement in sparsity across quantiles. Without loss of generality we assume the last $p-s$ components of $\bm \beta_{k}^*$ are zero for all considered quantiles. That is for all $k \in \{1,\ldots,K\}$, $\bm \beta_{k}^* = [\bar{\bm \beta}_{k}^{*\top},\mathbf{0}^\top]^\top$ and $\bar{\bm \beta}^*_{k} \in \Real^{s+1}$. 

In this section we discuss some of the theoretical properties of $\hat{\bm \beta}$, proofs are provided in the supplementary material. \citet{negahban2012} provided a general framework for proving consistency for penalized estimators. The results in that paper are not directly applicable to the proposed estimator because they rely on a differentiable loss function. However, many of the ideas presented in that paper can be used to understand the motivation behind some of the conditions and definitions used in this paper. Typically proofs of sparse penalized estimators rely on the difference between the estimator and the truth belonging to a cone. As outlined in \citet{negahban2012} this will rely on a relationship between sparsity tuning parameter $\lambda$ and the dual of the penalty evaluated at the gradient of the loss function. Due to the loss function being non-differentiable we deal with the subgradient $v_{jk} = n^{-1}\sum_{i=1}^n [\tau_k-I(\epsilon_i^k \leq 0)]x_{ij}$ with $\bm v_j = (v_{j1},\ldots,v_{jK})^\top \in \Real^{K}$. The dual of the penalty evaluated at this subgradient of the loss function is $\Lambda = \underset{j \in \{0,\ldots,p\}}{\max} \left|\left| \bm v_j \right|\right|_2$, and the cone is
\begin{equation*}
	\mathcal{C} = \left\{\bm \Delta \in \Real^{K(p+1)} \middle| \sum_{j \in \bar{\mathcal{S}}^c} ||\bm \Delta^j||_2 \leq 3 \sum_{j \in \bar{\mathcal{S}}} ||\bm \Delta^j||_2\right\}.
\end{equation*}
%

The following lemma provides a deterministic results about the difference between the estimator and the truth residing in $\mathcal{C}$ depending on the values of $\lambda$ and $\Lambda$.
%

\begin{lemma}
	\label{lemCone}
	If $L_1 \lambda  \geq  \Lambda$ and $L_1 < C_1$ then $\hat{\bm \beta}-\bm \beta^* \in \mathcal{C}$. 
\end{lemma}

Lemma \ref{lemCone} provides a basic result of most penalized sparse estimators and is an important results for the proof of consistency. Additional conditions are needed to prove consistency. 

\begin{condition}
	\label{condDens}
	Let $F_{i}$, $f_{i}$ and $f'_{i}$ be the cdf, pdf and first derivative of the pdf for $y_i$ given $\bm x_i$. There exists a positive constants $c_1$ and $c_2$ such that $\underset{k \in \{1,\ldots,K\}, i \in \{1,\ldots,n\}}{\inf} f_{i}(t) \geq c_1$ and $\underset{k \in \{1,\ldots,K\}, i \in \{1,\ldots,n\}}{\inf} |f'_{i}(t)| \leq c_2$ for all $t \in \Real$.
\end{condition}

\begin{condition}
	\label{condX2}
	There exists a positive constant $C_x$ such that $\underset{i \in \{1,\ldots,n\}\, j \in \{0,\ldots,p\}}{\max}|x_{ij}| < C_x$.
\end{condition}

\begin{condition}
	\label{condX}
	There exists positive constants $b_1,\ldots,b_K$, $B_1,\ldots,B_K$, $\tilde{b}$ and $\tilde{B}$ such that for all $k \in \{1,\ldots, K\}$, 
	$
	\tilde{b} ||\bm \Delta_k||_2^2 \leq b_k ||\bm \Delta_k||_2^2 \leq \underset{\bm \Delta_k \in \mathcal{C}_k, ||\bm \Delta_k||_2 \neq 0}{\inf}  \bm \Delta_k^\top E[\bm x_i\bm x_i^\top]\bm \Delta_k \leq \underset{\bm \Delta_k \in \mathcal{C}_k, ||\bm \Delta_k||_2 \neq 0}{\sup}  \bm \Delta_k^\top E[\bm x_i\bm x_i^\top]\bm \Delta_k \leq B_k||\bm \Delta_k||_2^2 \leq \tilde{B} ||\bm \Delta_k||_2^2.
	$
\end{condition}

\begin{condition}
	\label{condRest}
	For $\bm \Delta = (\bm \Delta_1,\ldots,\bm \Delta_K)$, where $\bm \Delta \in \mathcal{C}$, there exists positive constants $d_1,\ldots,d_K$ and $\tilde{d}$, such that for all $k \in \{1,\ldots,K\}$
	$
	\tilde{d} \leq d_k = \frac{c_1}{c_2\sqrt{B}_k} \underset{\bm \Delta \in \mathcal{C}, ||\bm \Delta_k||_2 \neq 0}{\inf} \frac{E\left[(\bm x_i^\top\bm \Delta_k)^2\right]^{3/2}}{E\left[|\bm x_i^\top\bm \Delta_k|^3\right]}.
	$
\end{condition}

All of the conditions are common for quantile regression or penalized estimators. Conditions similar to Condition \ref{condDens} have appeared in many quantile regression papers \citep{qr_lasso,scad_qr,qr_group_lasso}, and was used in Section 4.2 of \citet{qrBook} to prove asymptotic normality of the standard quantile regression estimator. Condition \ref{condX2} assumes that the support is bounded. For a model with heteroscedasticity the quantile coefficients will vary with the quantile, but the fits do not need to cross if the support is bounded. \citet{negahban2012} provide a useful discussion on why Condition \ref{condX} is needed for penalized estimators. \citet{qr_lasso} introduced a condition similar to Condition \ref{condRest} for quantile regression with lasso and similar conditions have been used for group lasso penalties with quantile regression \citep{qr_group_lasso,sherwood2022quantile}. These conditions are used to prove the consistency of $\hat{\bm \beta}$.

\begin{theorem}
	\label{bigTheorem}
	Define $\delta^* = 8C^*\sqrt{\frac{s+1}{n}}\left[4D_1\sqrt{\log(p)}+\sqrt{2K}\right]$ where $D_1$ and $D_2$ as positive constants and. Assume Conditions \ref{condDens}-\ref{condRest} hold and $\lambda=C_x\sqrt{\frac{K}{n}}+2C_x\sqrt{\frac{D_2 K \log(p)}{n}}$ then with probability at least $1-16p^{1-D_1}-2p^{1-D_2}$, 
	\begin{equation*}
		||\hat{\bm \beta}-\bm \beta^*||_2 \leq \frac{3}{c_1 \tilde{b}} \left[ \max\left(\delta^*, K\sqrt{\frac{8\tilde{B}}{n}}\right)+c_2\lambda \sqrt{s}\right],
	\end{equation*}
	provided the upper bound is smaller than $\tilde{d}$. 
\end{theorem}

The following corollary applies Theorem \ref{bigTheorem} to obtain a rate of convergence for the proposed estimator. 
\begin{corollary}
	\label{cor1}
	If assumptions of Theorem \ref{bigTheorem} hold, $p \rightarrow \infty$, $s \log(p) = o(n)$, $K=O(1)$, $D_1>1$, and $D_2 > 1$ then $||\hat{\bm \beta}-\bm \beta^*||_2 = O_P\left[\sqrt{\frac{sK\log(p)}{n}}\right]$.
\end{corollary}

Corollary holds because under the stated conditions, $1-16p^{1-D_1}-2p^{1-D_2} \rightarrow 1$, $\sqrt{\frac{sK\log(p)}{n}}=o(1)$, and dominates all other terms in the upper bound of Theorem \ref{bigTheorem}. Corollary \ref{cor1} provides the same rate of convergence as the estimator that replaces the group lasso penalty with an individual lasso penalty for each quantile-predictor combination \citep{qr_lasso}.

\section{Algorithm} \label{sec:algorithm}
Solving \eqref{eq:obj1} is challenging because both the loss and penalty functions are non-differentiable. It can be framed as a second-order cone programming problem but algorithms for solving these problems are typically slow for large values of $n$ or $p$.
It has been shown that using a differentiable Huber loss to approximate the nondifferentiable quantile loss can significantly improve the computational efficiency while preserving the estimation accuracy \citep{yi2017semismooth,sherwood2022quantile,Ouhourane2022}. Specifically, the Huber-approximated quantile loss is defined as
\begin{align}
	\rho_{\tau}(u)\approx\frac{1}{2}\left[h_{\gamma}(u)+(2\tau-1)u\right] \equiv \frac{1}{2}h_{\gamma}^{\tau}(u), \label{eq:haq}
\end{align}
where $h_{\gamma}(u)$ is the classical Huber loss defined as
\begin{equation*}
	h_{\gamma}(u)=\frac{u^2}{2\gamma}I(u \leq \gamma) + (|u|-\frac{\gamma}{2})I(|u|>\gamma).
\end{equation*}
Then the minimizer of the objective function \eqref{eq:obj1} can be approximated by the minimizer of the following objective function when $\gamma \to 0$, 
\begin{align}
	L^{\gamma}_{\lambda}(\bm \beta| \tau_1,...,\tau_K; \mathbf{X}, \mathbf{Y})=\frac{1}{2n}\sum_{k=1}^{K}\sum_{i=1}^{n}h_{\gamma}^{\tau_k}(y_i-\bm x_i^\top \bm \beta_k)+\lambda \sum_{j=1}^{p}\|\bm \beta^j\|_2. \label{eq:obj2}
\end{align}

Since the problem is framed as the quantile regression with group lasso penalty, we adopt the algorithm from \citet{sherwood2022quantile}, which is motivated by the groupwise-majorization-descent (GMD) algorithm \citep{yang2015fast}. To implement the algorithm, we augment the response and covariate matrix. 
Specifically, we construct an augmented data frame $(\bm Y_{aug}, \bm X_{aug})$, where $\bm Y_{aug}=(\bm Y^\top, \bm Y^\top,..., \bm Y^\top)^\top\in \mathbb{R}^{nK}$ and $\bm X_{aug}=\bm I_K \bigotimes \bm X \in \mathbb{R}^{nK \times pK}$ where $\bm I_K$ is a $K$-dimensional identity matrix. By reordering $\bm X_{aug}$ we obtain $[\bm I_K \bigotimes X_1 \quad \bm I_K \bigotimes X_2 \quad \cdots \quad \bm I_K \bigotimes X_p]$, so that $\bm I_K \bigotimes X_j$ can be regarded as a group of $K$ covariates for $j=1,...,p$, where $X_j=(x_{1j},...,x_{nj})$ is the $j$th column of the design matrix $\bm X$. 

To apply the GMD algorithm, we only need to verify that the unpenalized loss function, denoted $L^{\gamma}(\bm \beta)$ (i.e., the objective function \eqref{eq:obj2} without the penalty term), satisfies the quadratic majorization (QM) condition. That is, $L^{\gamma}(\bm \beta)$ is differentiable everywhere, which is trivial, and  for any $\bm \beta$ and $\tilde{\bm \beta}$, 
\begin{align}
	L^{\gamma}(\bm \beta)\le L^{\gamma}(\tilde{\bm \beta})+(\bm \beta- \tilde{\bm\beta})^\top\nabla L^{\gamma}(\tilde{\bm \beta})+\frac{1}{2}(\bm \beta-\tilde{\bm \beta})^\top \bm H (\bm \beta-\tilde{\bm \beta}), \label{eq:QM}
\end{align}
where $\bm H$ is a positive-definite matrix. Based on Lemma 1 of \citet{yang2015fast}, \eqref{eq:QM} holds if the first-order derivative of the loss function $h_{\gamma}^{\tau}$ is Lipschitz continuous, and under such a case $\bm H=\frac{2C}{n}\bm X_{aug}^\top \bm X_{aug}$ where $C$ is the Lipschitz constant. This has been verified by \citet{sherwood2022quantile} for their Huber-approximated loss function, which is the same as our $h_{\gamma}^{\tau}$. Specifically, the Lipschitz constant for $\frac{1}{2}h_{\gamma}^{\tau}$ is $\frac{1}{2\gamma}$, so that the positive definite matrix in \eqref{eq:QM} is $\bm H=\frac{1}{n\gamma}\bm X_{aug}^\top  \bm X_{aug}$. 

As detailed in \citet{sherwood2022quantile}, given the shrinkage parameter $\lambda$, the algorithm iteratively updates the coefficients of one group at a time. In our case, with the augmented data, the coefficient vector $\bm \beta^j$ of ``group" $j$ for $j\in\{0,1,...,p\}$, where $j=0$ indicates the intercepts, is updated by
\begin{align*}
	\tilde{\bm \beta}^0(new)&=\left\{-\gamma\left[\nabla L^{\gamma}(\tilde{\bm \beta})\right]^0+\tilde{\bm \beta}^0 \right\},\\
	\tilde{\bm \beta}^j(new)&=\frac{1}{\xi_j}\left\{-\gamma\left[\nabla L^{\gamma}(\bm \tilde{\bm \beta})\right]^j+\xi_j  \tilde{\bm\beta}^j\right\}\left(1-\frac{\lambda}{\left\|-\left[\nabla L^{\gamma}(\tilde{\bm \beta})\right]^j+\xi_j  \tilde{\bm \beta}^j \right\|_2}\right)_{+},
\end{align*}  
where $\left[\nabla L^{\gamma}(\bm \tilde{\bm \beta})\right]^j$ is the gradient of the unpenalized loss corresponded to ``group" $j$, i.e., the $j$th covariate, and $\xi_j$ is the largest eigenvalue of the sub-matrix $\bm H_j$. 
The algorithm updates $\tilde{\bm \beta}^0, \tilde{\bm \beta}^1, ..., \tilde{\bm \beta}^p$ sequentially and then iterates through the cycle till the entire coefficient vector $\tilde{\bm \beta}$ converges. The convergence is theoretically guaranteed due to \eqref{eq:QM} and the convexity of the loss function.

To obtain the full solution path with a sequence of $\lambda$'s while reduce the computation cost, we follow the warm start approach and adopt the strong rule \citep{tibshirani2012strong} that helps with pre-screening the signals. The implementation follows the algorithm in \citet{sherwood2022quantile}. The proposed algorithm is implemented in the CRAN package rqPen \citep{rqPen}.

\section{Simulation} \label{sec:simulation}
We conduct various numerical experiments that demonstrate the advantage of our proposed method in estimating quantile regression at multiple quantiles. We compare our proposed method (gQ-Lasso) to quantile regression with Lasso penalty (rq-Lasso) \citep{qr_lasso} and the globally concerned quantile regression (GCQR) \citep{zheng2015globally}. The rq-Lasso estimates the model identify important variables for each given quantile, while GCQR jointly estimates multiple quantiles and identifies the set of globally relevant variables over a set of quantiles by utilizing a weighted $L_1$ penalty. Specifically, let $\bm \tilde{\beta}^j$ be initial estimates of the $j$th coefficients by using quantile regression with lasso. GCQR minimizes,
\begin{equation}
	\label{gcqr}
	\frac{1}{n}\sum_{k=1}^{K}\sum_{i=1}^{n}\rho_{\tau_k}(y_i- \bm x_i^\top \bm \beta_k)+\lambda \sum_{j=1}^{p} \frac{1}{\|\bm \tilde{\beta}^j\|_{\infty}}\sum_{k=1}^K |\beta_{jk}|.
\end{equation}
\citet{zheng2015globally} propose two other weights to be used in an adaptive lasso-type penalty. 
Note that neither of the two alternatives guarantees consistent selection across all quantiles, although GCQR encourages so. In addition to creating unnecessary complexity in interpretation, such inconsistency, as we demonstrate, can cause severe crossing quantiles. Nevertheless, we are aware some existing methods that are designed for consistent variable selection. For instance, \citet{ZOU20085296} proposed to use an infinity norm penalty and framed the problem to a linear programming. At the end of this section, we show that such an approach suffers from significant computation burden.

We generate data from the following model: $Y=0.5X_1-X_2+1.5X_3-2X_4 +\epsilon$,
where $X_1 \sim Pois(2)$; $X_2\sim U(1,5)$; and the rest $X \sim N(0, \Sigma)$, where the off-diagonal elements of $\Sigma$ decay as $0.3^{|i-j|}$. We consider five different types of error:
\begin{enumerate}
	\item (normal error): $\epsilon^{(1)} \sim N(0,1)$;
	\item (heavy-tailed error): $\epsilon^{(2)} \sim t(2)$;
	\item (error with heteroscedasticity): $\epsilon^{(3)}=\bm \zeta^\top \bm x\epsilon^*$, where $\bm \zeta=(1,1,1.5,0,0,...,0)^\top \in \mathbb{R}^{p+1}$, $\epsilon^* \sim N(0,1)$ and the first entry of $\bm x$ is the constant term $1$;
	\item (asymmetric error): $\epsilon^{(4)} \sim \chi^2(3)-3$;
	\item (heteroscedastic and asymmetric): $\epsilon^{(3)}+\epsilon^{(4)}$.
\end{enumerate}
Note that the effective true coefficients varies across different quantiles. Consider the general location-scale model \eqref{locScale}. For a given $\tau_k$ the true coefficients can be written as $\bm \beta_{k}^*=\bm \beta^*+\bm \zeta F^{-1}(\tau_k)$, where $F(\cdot)$ is the cumulative distribution function of $\epsilon$. Throughout our simulation, we estimate models for $\tau \in \{0.1, 0.3, 0.5, 0.7, 0.9\}$.  

We consider different sample sizes for which $N \in \{200, 500, 1000\}$ and $p \in \{20, 50, 100\}$. For each scenario, we generate 100 random samples and assess the model performance with respect to the following metrics: (1) average model error (ME) across quantiles, where ME for quantile $\tau_k$ is defined as $\text{ME}_{k}=(\bm \beta_k^*-\hat{\bm \beta}_k)^\top \bm X^\top \bm X (\bm \beta_k^*-\hat{\bm \beta}_k)$; (2) the number of cross quantiles (NCQ) in fitted values, which is defined as $\sum_{i=1}^{N} (1-I(\hat{y}_{i,\tau_1} \le \hat{y}_{i,\tau_2} \le \cdots \le \hat{y}_{i,\tau_k}))$; (3) average false positive (FP); and (4) false negative (FN) of variable selection across quantiles. The optimal value of $\lambda$ for each method is chosen via a 5-fold cross-validation that minimizes the respective objective function.

Figure \ref{fig:simuME} and \ref{fig:simuNCQ} compare the proposed gQ-Lasso with rq-Lasso and GCQR with respect to model error (ME) and the number of cross quantiles (NCQ). For ME, the proposed gQ-Lasso outperforms rq-Lasso across almost all simulation settings, while GCQR delivers better performance under homoscedastic errors but underperforms gQ-Lasso when the data has more complex errors (e.g., the heteroscedastic plus asymmetric error) and the sample size is relatively small. As the sample size $N$ increases, the ME decreases for all methods, and the difference becomes negligible. 
As expected, the ME increases as the data become noisier. 
Regarding crossing-quantiles, measured by NCQ, Figure \ref{fig:simuNCQ} depicts large discrepancies between gQ-Lasso and the other two methods. One can see that the issue of crossing-quantile for rq-Lasso is very severe when sample size is moderate. GCQR has smaller number of crossing quantiles than rq-Lasso under homoscedastic errors but it becomes the worst when heteroscedasticity is introduced to the data. These results provide evidences that the crossing quantile issue can be largely addressed by the proposed gQ-Lasso due to the property of consistent selection across quantiles. 

\begin{figure}[h!]
	\centering
	\includegraphics[scale=0.6]{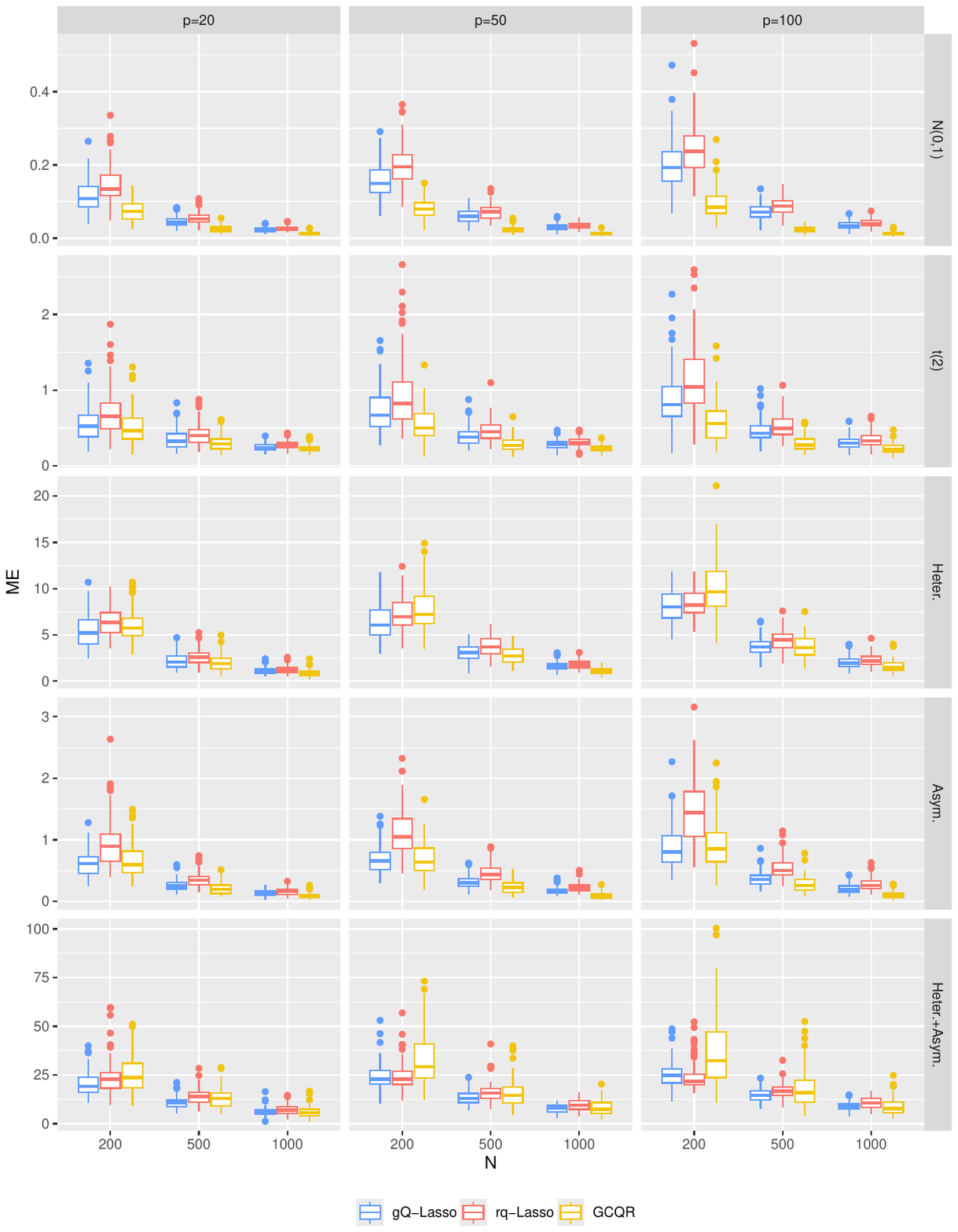}
	\caption{Boxplot of model error (ME) of gQ-Lasso and rq-Lasso} \label{fig:simuME}
\end{figure}

\begin{figure}[h!]
	\centering
	\includegraphics[scale=0.6]{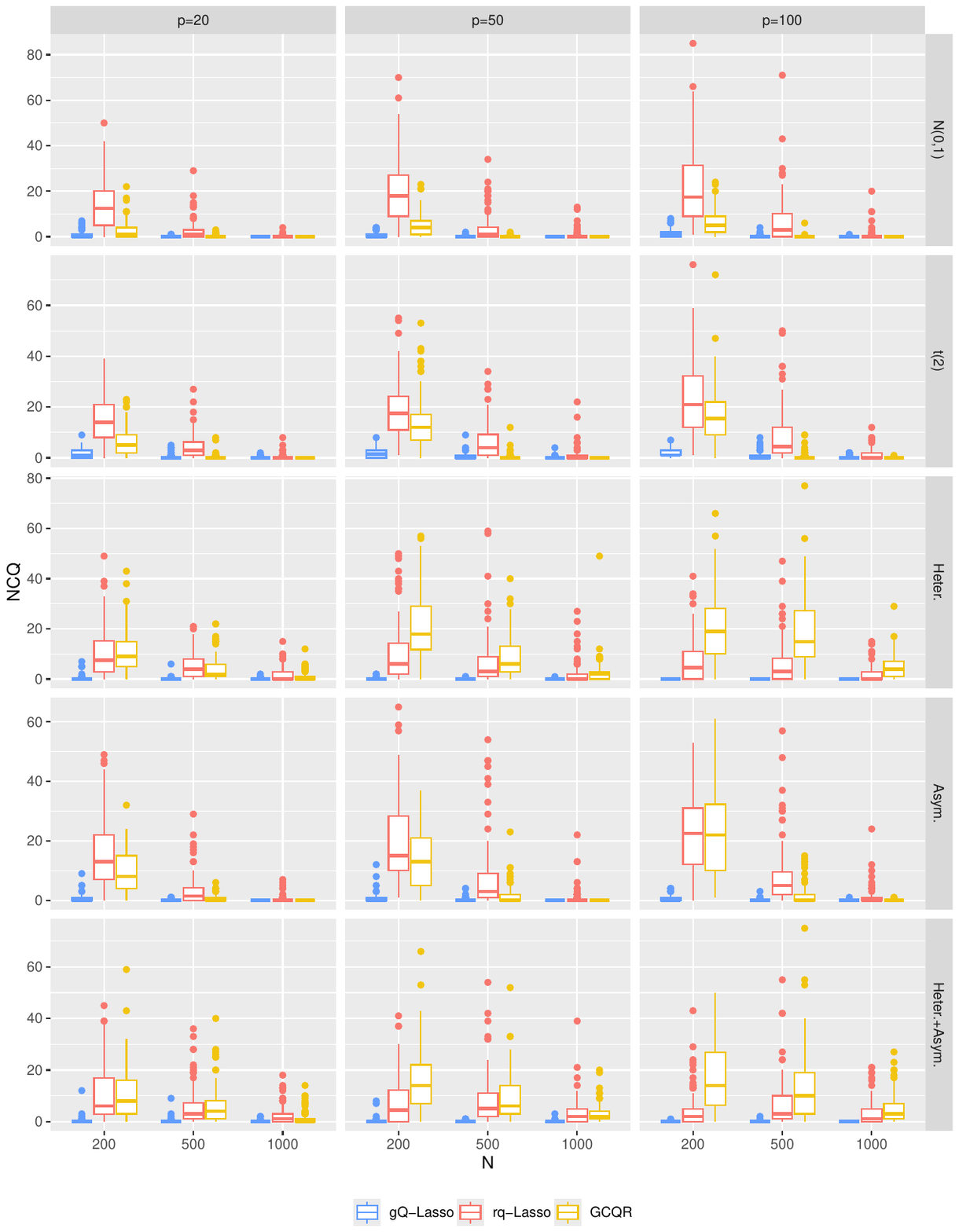}
	\caption{Boxplot of the number of cross-quantiles (NCQ) of gQ-Lasso and rq-Lasso} \label{fig:simuNCQ}
\end{figure}

The results regarding variable selection in terms of false positive (FP) and false negative (FN) are reported in Table S1 and S2 in the supplemental material. In general gQ-Lassso tends to pick a larger model due to the group penalty, and thus has more false positives but less false negatives than the other methods.  

We also compare gQ-Lasso with the approach from \citet{ZOU20085296} which we call multi-quantile regression with $L_\infty$ penalty (mrq-$L_\infty$).  Such an infinity norm penalty facilitates the estimator to achieve consistent variable selection. Specifically, their estimator is 
\begin{align*}
	\hat{\bm \beta}=\argmin_{\bm \beta}\sum_{k=1}^{K}\sum_{i=1}^{n}\rho_{\tau_k}(y_i- \bm x_i^\top \bm \beta_k)+\lambda \|\bm \beta\|_{1\infty},
\end{align*}
where $\|\bm \beta\|_{1\infty}=\sum_{j=1}^{p}\max_k{|\beta_{jk}|}$. This optimization problem is formulated as a linear programming, which we show in Table \ref{tab:comparison-time} that it incurs a substantial computational burden. Due to this reason, the comparison was conducted only based on 10 replicates for $N\in \{200, 500\}$ and $p \in {20, 100}$. Results for model error and number of crossing-quantiles are reported in Table S3 and S4 in supplemental materials. The results show that mrq-$L_\infty$ slightly outperforms gQ-Lasso under $p=20$, but substantially worse for $p=100$.

\begin{table}[htbp]
	\centering
	\caption{Computational time (second) of gQ-Lasso and mrq-$L_\infty$ based on 10 replicates.}
	\label{tab:comparison-time}
	\resizebox{\textwidth}{!}{%
		\footnotesize
		\begin{tabular}{llcccccccccc}
			\toprule
			& & \multicolumn{2}{c}{\( N(0,1) \)} & \multicolumn{2}{c}{\( t(2) \)} & \multicolumn{2}{c}{Heter.} & \multicolumn{2}{c}{Asym.} & \multicolumn{2}{c}{Heter.+Asym.} \\
			\cmidrule(lr){3-4} \cmidrule(lr){5-6} \cmidrule(lr){7-8} \cmidrule(lr){9-10} \cmidrule(lr){11-12}
			$N$ & $p$ & {gQ-Lasso} & {mrq-$L_\infty$} & {gQ-Lasso} & {mrq-$L_\infty$} & {gQ-Lasso} & {mrq-$L_\infty$} & {gQ-Lasso} & {mrq-$L_\infty$} & {gQ-Lasso} & {mrq-$L_\infty$} \\
			\midrule				 
			200 & 20 & 5.99 & 258.32 & 10.06 & 255.65 & 15.41 & 248.59 & 10.23 & 264.77 & 16.81 & 248.39 \\
			& & \textit{(0.57)} & \textit{(4.27)} & \textit{(1.05)} & \textit{(3.12)} & \textit{(1.18)} & \textit{(4.88)} & \textit{(1.22)} & \textit{(20.15)} & \textit{(0.73)} & \textit{(2.69)} \\
			& 100 & 34.30 & 2084.20 & 51.98 & 2067.93 & 92.87 & 1960.75 & 63.77 & 2064.94 & 98.43 & 1768.16 \\
			& & \textit{(4.68)} & \textit{(42.27)} & \textit{(5.16)} & \textit{(26.08)} & \textit{(6.30)} & \textit{(72.55)} & \textit{(4.35)} & \textit{(104.84)} & \textit{(5.60)} & \textit{(46.13)} \\
			\midrule
			500 & 20 & 14.95 & 1265.16 & 27.40 & 1267.83 & 43.83 & 1250.25 & 32.31 & 1342.90 & 61.54 & 1265.15 \\
			& & \textit{(2.79)} & \textit{(22.64)} & \textit{(3.30)} & \textit{(25.82)} & \textit{(8.20)} & \textit{(22.89)} & \textit{(3.28)} & \textit{(58.13)} & \textit{(5.96)} & \textit{(32.49)} \\
			& 100 & 38.70 & 6254.49 & 93.34 & 6196.01 & 161.59 & 5304.61 & 105.50 & 5795.98 & 190.03 & 5233.57 \\
			& & \textit{(4.50)} & \textit{(267.17)} & \textit{(17.32)} & \textit{(218.20)} & \textit{(13.41)} & \textit{(100.94)} & \textit{(11.93)} & \textit{(139.21)} & \textit{(13.08)} & \textit{(135.51)} \\
			\bottomrule
		\end{tabular}
	}
\end{table}

Overall, our simulation studies demonstrate the advantage of the proposed gQ-Lasso in ME, NCQ and FN, while the results of FP indicate that gQ-Lasso tends to over select variables, which is expected due to the property of consistent selection across quantiles. Such a property significantly mitigates the issue of crossing-quantiles which can otherwise lead to confusing and unreliable quantile predictions. The comparison between gQ-Lasso and mrq-$L_\infty$ demonstrates that gQ-Lasso offers a significant advantage in computational efficiency. 

\section{Equity Premium Prediction}\label{sec:empirical}

\subsection{Data}
We re-examine an important finance problem, equity premium predictability, using our proposed method, and compare with several alternative approaches. 
We focus on the one-month-ahead return prediction, where the response variable is the monthly excess return of the S\&P 500 index, defined as the value weighted return with dividends minus the risk-free rate. For predictors, we directly use the market and economic variables from \citet{goyal2023comprehensive} who has made the data publicly available.\footnote{\url{https://drive.google.com/file/d/17eRrT9yVYmnS-ReJJJhN2RY2llldpLK7/view?usp=drive_link}} Among all 46 predictors, we select only monthly-recorded variables and further exclude those with more than 20\% missing values. This results in 24 predictors and $N=1141$ observations spanning from December 1926 to December 2021. We lag predictors by a month so that the estimated model is a month-ahead predictive model. Table \ref{tab:summarystat} provides a brief description of each variable and their summary statistics. Following exploratory analysis, we take logarithm transformations of \textit{stock variance (svar)}, \textit{default yield spread (dfy)}, and \textit{return dispersion (rdsp)}.

\begin{table}[h]
	\centering
	\caption{Variable description and summary statistics} \label{tab:summarystat}
		\footnotesize
		\begin{tabular}{llrrrrr}
			\hline
			Variable & Description & Mean & Std & Min & Median & Max\\
			\hline
			ExRet & Excess return of S\&P500 index & 0.007 & 0.054 & -0.288 & 0.010 & 0.414\\
			avgcor & average correlation of daily return & 0.284 & 0.121 & 0.035 & 0.267 & 0.708 \\ 
			b/m & Book to market & 0.554 & 0.269 & 0.121 & 0.531 & 2.028 \\ 
			d/e & Dividend payout ratio & -0.643 & 0.326 & -1.244 & -0.633 & 1.380 \\	 
			d/p & Dividend price ratio & -3.401 & 0.473 & -4.524 & -3.371 & -1.873 \\ 
			d/y & Dividend yield & -3.396 & 0.470 & -4.531 & -3.362 & -1.913 \\
			dfr & Default return spread & 0.000 & 0.014 & -0.098 & 0.000 & 0.074 \\
			dfy & Default yield spread & 0.011 & 0.007 & 0.003 & 0.009 & 0.056 \\ 		
			dtoat & nearness to Dow lifetime high & 0.785 & 0.236 & 0.112 & 0.889 & 1.000 \\
			dtoy & nearness to Dow 52-week high & 0.919 & 0.105 & 0.276 & 0.957 & 1.000 \\ 
			e/p & Earning price ratio & -2.758 & 0.421 & -4.836 & -2.808 & -1.775 \\	
			fbm & single factor from B/M cross-section & 0.173 & 0.206 & -1.418 & 0.169 & 0.693 \\								
			infl & Inflation & 0.002 & 0.005 & -0.021 & 0.002 & 0.059 \\ 
			ltr & Long term return & 0.005 & 0.025 & -0.112 & 0.003 & 0.152 \\
			lty & Long term yield & 0.050 & 0.028 & 0.006 & 0.041 & 0.148 \\
			lzrt & 9 illiquidity measures & -1.695 & 0.335 & -4.722 & -1.668 & -0.962 \\  
			ntis & Net equity expansion & 0.016 & 0.026 & -0.056 & 0.016 & 0.177 \\ 
			ogap & output gap of industrial production & -0.004 & 0.133 & -0.562 & -0.009 & 0.384 \\ 
			rdsp & return dispersion & 0.039 & 0.029 & 0.015 & 0.031 & 0.420 \\	
			skvw & average stock skewness & 0.033 & 0.054 & -0.379 & 0.036 & 0.398 \\
			svar & Stock variance & 0.003 & 0.006 & 0.000 & 0.001 & 0.073 \\ 			
			tail & tail-risk from cross-section & 0.425 & 0.058 & 0.286 & 0.438 & 0.589 \\
			tbl & Treasury-bill rate & 0.033 & 0.031 & 0.000 & 0.028 & 0.163 \\		
			tms & Term spread & 0.017 & 0.013 & -0.036 & 0.017 & 0.045 \\ 				
			wtexas & oil price change & 0.005 & 0.075 & -0.542 & 0.000 & 0.884 \\ 
			\hline
		\end{tabular}
\end{table}

\subsection{Model estimation based on full sample}
Our empirical analysis employs several existing methods for model estimation and variable selection, including both quantile and mean regressions for benchmark. Specifically, we compare three quantile regressions: (1) the proposed method (gQ-Lasso), (2) quantile regression with Lasso penalty (rq-Lasso), and (3) the globally concerned quantile regression with adaptive Lasso penalty (GCQR), and three mean regressions: (4) ordinary least square with stepwise selection using AIC (stepAIC), (5) classical Lasso model (Lasso), and (6) Huber regression with Lasso penalty (H-Lasso). For quantile methods, we model 9 quantiles where $\tau \in \{0.1, 0.2, ..., 0.9\}$. We use the 10-fold cross-validation to determine the optimal $\lambda$ values for the all penalized methods. 

We use heat maps as shown in Figure \ref{fig:coefplot} to show the estimated model coefficients and compare them across different models that are fitted based on the full sample period (1926/12 -- 2021/12). Unselected variables are marked using solid black color while nonzero coefficients are represented using color scale from blue to red. Comparing across three penalized quantile methods, the advantage of the proposed gQ-Lasso is that a variable is either in or out across all quantiles, so that it provides a coherent and easy-to-interpret model. In contrast, for rq-Lasso and GCQR, variables often jump in and out of the model across quantiles without any clear patterns. This is extremely inconvenient for practitioners to interpret the model. For instance, in GCQR, the variable ``dtoat" has a positive effect at 0.1 quantile, then the effect disappear from 0.2 to 0.5, and the signal goes back and forth afterwards. Note that the idea of GCQR is to recover the true signals in nearby quantiles as the effect is often expected to remain the same in nearby quantiles. Yet, the results of GCQR shown in Figure \ref{fig:coefplot} do not align with its motivation. Several variables such as \textit{avgcor}, \textit{dfr}, \textit{fbm}, and \textit{infl} are among the strongest predictive variables that are agreed across all methods, while others, such as \textit{dtoy}, \textit{e/p}, \textit{skvw}, and \textit{svar} show strong signals under quantile regressions but none of them are picked by the three mean regressions. The variable \textit{b/m} shows a very strong predictive power under OLS and Lasso, but it is unselected by Huber-Lasso and a majority of the quantile regression models. This is potentially due to extreme values as Huber and quantile regressions are known for their robustness. Another interesting finding is on variable \textit{log(svar)}, which is ignored by all three mean regressions but picked by quantile regressions as it clearly shows strong effect on lower and upper quantiles as shown in Figure \ref{fig:motivation2}.

\begin{figure}[h!]
	\centering
	\includegraphics[scale=0.5]{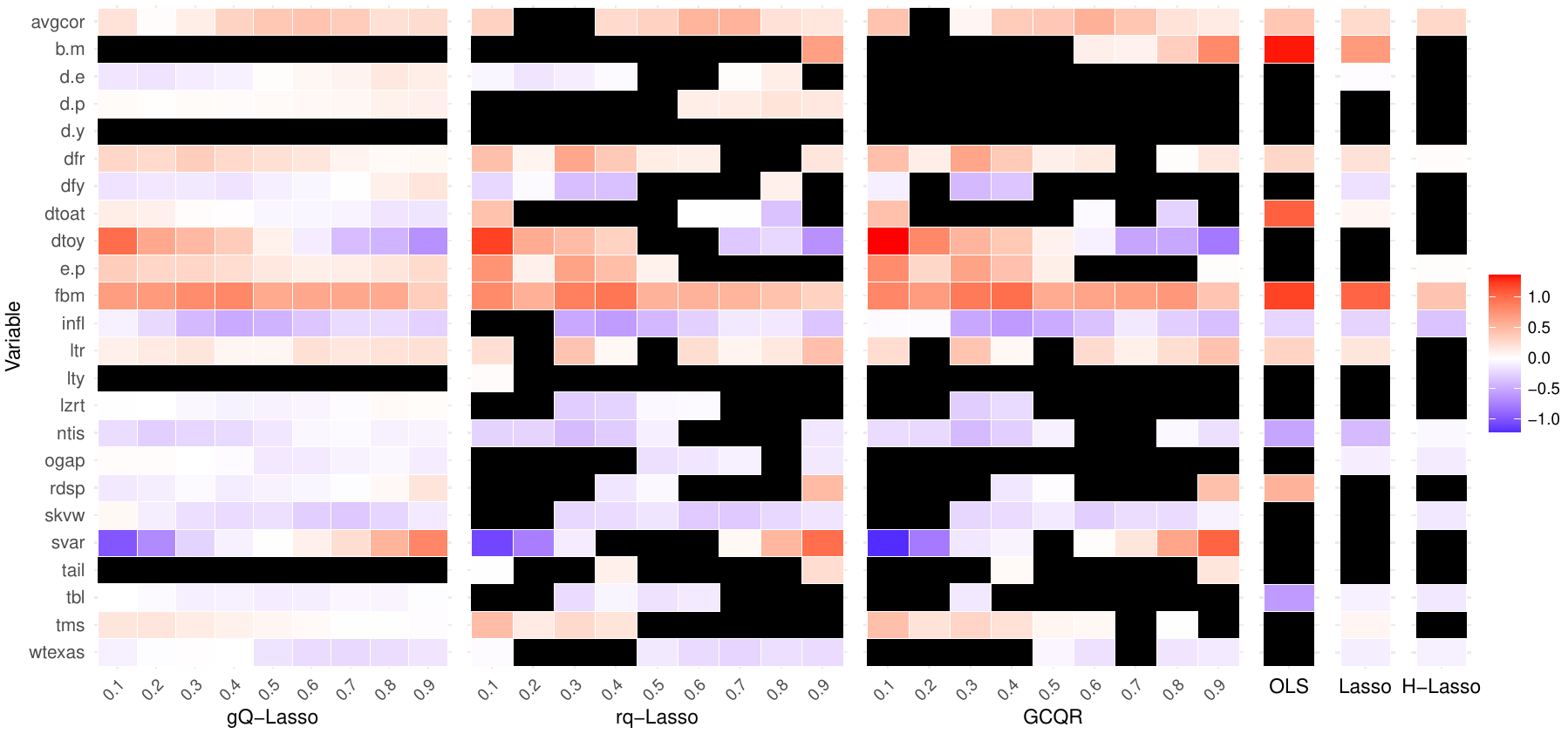}
	\caption{Heatmap of estimated coefficients. Black entries represent a coefficient being set to zero.} \label{fig:coefplot}
\end{figure}

To take into account selection uncertainty, we conduct 1000 bootstrap sampling and compute the percentage of selection for each variable across each method. These results are reported in Table S5 in the supplemental material. 

\subsection{Assessing in- and out-of-sample prediction}	
We examine the predictability of excess return via several measures under both full sample (in-sample) and out-of-sample analysis. The latter is conducted through the standard expanding window forecast approach. Specifically, we use all available data up to time stamp $T$ to estimate the model and predict the excess return at time $T+1$, then re-estimate the model using data up to time $T+1$ and predict the excess return at time $T+2$, and keep it going till the end of the prediction window. 
To assess the prediction accuracy for both mean and quantiles. For mean prediction, we follow the literature \citep{welch2008comprehensive,rapach2010out} and compute the $R^2$ for both in-sample and out-of-sample predictions. The in-sample $R^2$ is defined as
\begin{align}
	R^2_{IS}=1-\frac{\sum_{t=1}^{N}(Y_t-\hat{Y}_t)^2}{\sum_{t=1}^{N}(Y_t-\bar{Y})^2}, \label{eq:RsqIS}
\end{align}
where $\hat{Y}_t$ is the predicted return at time $t$ and $\bar{Y}$ is the average of return computed from the full sample. Additionally, the out-of-sample $R^2$ is defined as
\begin{align}
	R^2_{OOS}=1-\frac{\sum_{OS_1\le t\le OS_n}(Y_t-\hat{Y}_t)^2}{\sum_{OS_1\le t\le OS_n}(Y_t-\bar{\bar{Y}}_{t-1})^2}, \label{eq:RsqOOS}
\end{align}
where $OS_1$ and $OS_N$ represent the first and last time stamp of the prediction period and $\bar{\bar{Y}}_{t-1}$ is the historic mean return up to time point $t-1$. A positive $R^2_{OOS}$ implies that the model prediction beats historic mean. To obtain the mean prediction from quantile regressions, we compute the equal-weighted average across all quantile predictions. Hence, the $R^2$ defined in \eqref{eq:RsqIS} and \eqref{eq:RsqOOS} can be directly applied.

To assess quantile predictions, we modify \eqref{eq:RsqIS} and \eqref{eq:RsqOOS} by first replacing $\bar{Y}$ and $\bar{\bar{Y}}$ with $\bar{Y}^{(\tau)}$ and $\bar{\bar{Y}}^{(\tau)}$, respectively, where the latter indicates historic $\tau$th quantile; and then replacing the square error with the quantile check loss \citep{qrgoodfit}. We denote $R^2_{mean}$ and $R^2_{\tau}$ as the R squares for mean and quantile predictions. Quantile prediction from mean regressions are obtained by assuming the errors are normally distributed with variance being that of the residuals. 

In addition to prediction accuracy measured by $R^2$, we also evaluate the accuracy of sign prediction, which is crucial for investor's decision-making. Quantile regression at the median has a particular advantage in assessing sign prediction. This is because the median is not necessarily equivalent to the mean, and $P(Y>0)\ge 0.5$ if $\hat{Y}_{\tau=0.5}>0$ and $P(Y>0)< 0.5$ otherwise. Define the Percentage Error of Sign Prediction (PES) as the proportion of incorrect sign predictions. For quantile regression, PES is computed as $\frac{1}{n}\sum_{i=1}^{n}I[\text{sign}(Y_i)\ne\text{sign}(\hat{Y}_{i,\tau=0.5})]$, and for mean regression, $\hat{Y}_{i,\tau=0.5}$ is replaced by the predicted value $\hat{Y}_i$.

Last but not the least, we assess the issue of cross-quantile prediction for the two quantile regression approaches. We compute the percentage of cross quantile predictions (PCQ), which is defined as $\text{PCQ}=\frac{1}{N}\sum_{i=1}^{N} \left[1-I(\hat{Y}_{i,\tau_1} \le \hat{Y}_{i,\tau_2} \le \cdots \le \hat{Y}_{i,\tau_k})\right]$. A higher PCQ indicates more problematic predictions, which suggests an abnormal behavior of the model, and investors should avoid.

Table \ref{tab:perform1} shows results for the full sample analysis with 1000 bootstrap where the standard errors are displayed in the parentheses. Firstly, among the three quantile regression approaches, gQ-Lasso significantly outperforms the rq-Lasso and GCQR in terms of PCQ, with the latter two being approximately 3-4 times higher. This result aligns with expectations and is consistent with findings from our simulation studies. It provides further evidence, supported by real data analysis, that consistent variable selection across quantiles can effectively mitigate the issue of quantile crossing. Furthermore, the proposed gQ-Lasso achieves the lowest PES among all methods, significantly outperforming the three mean regression approaches.
Regarding $R^2$, the three quantile regression approaches exhibit significantly higher values for quantile predictions ($R^2_{\tau}$) but lower values for mean predictions ($R^2_{mean}$) compared to the squared-error-based methods (stepAIC and Lasso). This is not surprising as $R^2_{mean}$ and $R^2_{\tau}$ correspond to the squared and check loss, respectively, which are the primary objective functions of the mean and quantile regressions in our analysis. The Huber-Lasso performs worse than all other methods in terms of $R^2$.

\begin{table}[h]
	\centering
	\caption{In-sample $R^2$ for mean and quantile predictions, percentage error of sign prediction (PES) and percentage crossing-quantile (PCQ). Bootstrap standard errors are in parentheses.} \label{tab:perform1}
	\resizebox{\textwidth}{!}{%
		\footnotesize
		\begin{tabular}{lrrrrrr}
			\hline
			& gQ-Lasso & rq-Lasso & GCQR & stepAIC & Lasso & H-Lasso\\
			\hline
			$R^2_{mean}$ & 0.074 \textit{(0.025)} & 0.071 \textit{(0.022)} & 0.074 \textit{(0.024)} & 0.108 \textit{(0.031)} & 0.101 \textit{(0.039)} & 0.049 \textit{(0.024)}\\
			\hline
			$R^2_{\tau}$ &&&&&&\\
			\cline{1-2}
			$\tau=0.1$ & 0.161 \textit{(0.030)} & 0.168 \textit{(0.028)} & 0.165 \textit{(0.029)} & 0.040 \textit{(0.032)} & 0.033 \textit{(0.032)} & 0.010 \textit{(0.031)} \\
			$0.2$ & 0.094 \textit{(0.020)} & 0.093 \textit{(0.022)} & 0.092 \textit{(0.023)} & 0.024 \textit{(0.024)} & 0.018 \textit{(0.024)} & 0.009 \textit{(0.022)} \\
			$0.3$ & 0.067 \textit{(0.016)} & 0.066 \textit{(0.017)} & 0.065 \textit{(0.016)} & 0.022 \textit{(0.019)} & 0.019 \textit{(0.018)} & 0.020 \textit{(0.016)} \\
			$0.4$ & 0.053 \textit{(0.013)} & 0.050 \textit{(0.015)} & 0.050 \textit{(0.014)} & 0.020 \textit{(0.016)} & 0.020 \textit{(0.015)} & 0.034 \textit{(0.014)} \\
			$0.5$ & 0.046 \textit{(0.011)} & 0.043 \textit{(0.013)} & 0.043 \textit{(0.013)} & 0.016 \textit{(0.014)} & 0.019 \textit{(0.013)} & 0.042 \textit{(0.013)} \\
			$0.6$ & 0.051 \textit{(0.012)} & 0.048 \textit{(0.011)} & 0.048 \textit{(0.012)} & 0.017 \textit{(0.014)} & 0.023 \textit{(0.013)} & 0.030 \textit{(0.012)} \\
			$0.7$ & 0.063 \textit{(0.015)} & 0.059 \textit{(0.013)} & 0.059 \textit{(0.015)} & 0.014 \textit{(0.015)} & 0.020 \textit{(0.014)} & 0.012 \textit{(0.014)} \\
			$0.8$ & 0.090 \textit{(0.019)} & 0.088 \textit{(0.017)} & 0.088 \textit{(0.018)} & 0.006 \textit{(0.020)} & 0.011 \textit{(0.019)} & -0.011 \textit{(0.020)} \\
			$0.9$ & 0.146 \textit{(0.029)} & 0.148 \textit{(0.029)} & 0.147 \textit{(0.030)} & 0.005 \textit{(0.030)} & 0.005 \textit{(0.030)} & -0.037 \textit{(0.027)} \\
			\hline
			PES & 0.362 \textit{(0.015)} & 0.363 \textit{(0.018)} & 0.366 \textit{(0.017)} & 0.591 \textit{(0.021)} & 0.403 \textit{(0.022)} & 0.635 \textit{(0.017)} \\
			PCQ  & 3.956 \textit{(2.978)} & 15.665 \textit{(5.239)} & 13.113 \textit{(4.804)} & - & - & - \\
			\hline
		\end{tabular}
	}
\end{table}

Moving to out-of-sample prediction, we compute the same metrics as in Table \ref{tab:perform1}. We set the full forecasting window from January 1965 to December 2021, and we use the expanding window approach to predict the monthly excess return one at a time. Additionally, we examine for several sub-prediction windows including 1965-1972, 1976-2007, and 2010-2019. These periods are chosen by excluding major recessions, e.g., the Oil Shock during 1973-1975, the great recession during 2008-2009, and COVID-19 in 2020, which have brought substantial market volatility. Table \ref{tab:perform2} shows the results for quantile regression (Panel A) and mean regression (Panel B) for the four different evaluation periods. Note that rq-Lasso and GCQR have significantly higher PCQs than our proposed method across all four evaluation periods. To further illustrate this issue, Figure \ref{fig:cq} displays the predicted excess returns generated by gQ-Lasso and rq-Lasso across all nine quantiles during the period 1965-1972. It is evident that rq-Lasso frequently exhibits line-crossing over time, whereas gQ-Lasso predictions show significantly fewer instances of such crossings. This example highlights that quantile predictions using rq-Lasso may be problematic.

\begin{table}[h]
	\centering
	\caption{Out-of-sample comparison for multiple prediction periods.} \label{tab:perform2}
	\resizebox{\textwidth}{!}{%
		\footnotesize
		\begin{tabular}{lrrrrrrrrrrrr}
			\hline
			&\multicolumn{3}{c}{1965-2021} & \multicolumn{3}{c}{1965-1972} & \multicolumn{3}{c}{1976-2007} & \multicolumn{3}{c}{2010-2019}\\
			\cline{2-4} \cline{5-7} \cline{8-10} \cline{11-13}
			&\multicolumn{12}{c}{Panel A: Quantile regression}\\
			& gQ-Lasso & rq-Lasso & GCQR & gQ-Lasso & rq-Lasso & GCQR & gQ-Lasso & rq-Lasso & GCQR & gQ-Lasso & rq-Lasso & GCQR\\ 
			\hline			
			$R^2_{mean}$ & -0.003 & -0.003 & -0.008 & 0.038 & 0.023 & 0.034 & 0.034 & 0.031 & 0.022 & 0.008 & 0.010 & 0.015\\ 
			\hline
			\multicolumn{2}{l}{$R^2_{\tau}$} &&&&&&&&&&&\\
			\cline{1-2}
			$\tau=0.1$      & 0.049 & 0.030 & 0.034 & 0.033 & 0.052 & 0.054 & 0.056 & 0.039 & 0.040 & 0.057 & 0.026 & -0.001 \\
			0.2      & 0.008 & 0.002 & 0.002 & 0.018 & 0.010 & 0.002 & 0.016 & 0.010 & 0.003 & 0.025 & 0.017 & 0.023 \\
			0.3      & -0.005 & -0.006 & -0.012 & 0.007 & -0.001 & 0.009 & 0.017 & 0.007 & 0.002 & -0.016 & 0.004 & 0.004 \\
			0.4      & -0.002 & 0.003 & -0.003 & 0.014 & -0.005 & 0.023 & 0.016 & 0.015 & 0.003 & -0.026 & -0.021 & -0.027 \\
			0.5      & 0.002 & 0.005 & 0.004 & 0.042 & 0.016 & 0.032 & 0.010 & 0.010 & 0.010 & -0.013 & -0.010 & -0.004 \\
			0.6      & 0.016 & 0.013 & 0.010 & 0.041 & 0.032 & 0.045 & 0.026 & 0.019 & 0.019 & -0.003 & -0.003 & -0.013 \\
			0.7      & 0.023 & 0.025 & 0.021 & 0.044 & 0.050 & 0.046 & 0.027 & 0.029 & 0.021 & 0.025 & 0.017 & 0.024 \\
			0.8      & 0.059 & 0.055 & 0.044 & 0.131 & 0.138 & 0.100 & 0.040 & 0.040 & 0.032 & 0.083 & 0.080 & 0.062 \\
			0.9      & 0.101 & 0.042 & 0.066 & 0.218 & 0.042 & 0.121 & 0.060 & 0.019 & 0.036 & 0.185 & 0.142 & 0.155 \\
			\hline 
			PES & 41.520 & 41.080 & 42.250 & 44.790 & 43.750 & 45.830 & 41.930 & 42.190 & 42.450 & 29.170 & 29.170 & 29.170 \\ 
			PCQ & 1.900 & 31.870 & 22.080 & 0.000 & 29.170 & 16.670 & 2.080 & 40.620 & 28.120 & 0.830 & 6.670 & 10.830 \\ 
			\hline
			&\multicolumn{12}{c}{Panel B: Mean regression}\\
			& stepAIC & Lasso & H-Lasso & stepAIC & Lasso & H-Lasso & stepAIC & Lasso & H-Lasso & stepAIC & Lasso & H-Lasso \\ 
			\hline
			$R^2_{mean}$ & -0.128 & -0.031 & -0.014 & -0.104 & 0.027 & 0.022 & -0.136 & -0.017 & 0.014 & -0.024 & 0.009 & -0.040\\
			\hline
			\multicolumn{2}{l}{$R^2_{\tau}$} &&&&&&&&&&&\\
			\cline{1-2}
			$\tau=0.1$ & -0.102 & -0.077 & -0.054 & -0.184 & -0.180 & -0.116 & -0.141 & -0.090 & -0.068 & 0.038 & -0.031 & -0.005 \\ 
			0.2 & -0.115 & -0.087 & -0.051 & -0.133 & -0.110 & -0.041 & -0.174 & -0.115 & -0.071 & 0.033 & -0.053 & -0.028 \\ 
			0.3 & -0.098 & -0.061 & -0.030 & -0.146 & -0.084 & -0.028 & -0.147 & -0.080 & -0.040 & 0.039 & -0.027 & -0.025 \\ 
			0.4 & -0.071 & -0.032 & 0.001 & -0.131 & -0.033 & 0.012 & -0.096 & -0.038 & 0.002 & -0.001 & -0.027 & -0.023 \\ 
			0.5 & -0.059 & -0.025 & 0.007 & -0.038 & 0.030 & 0.033 & -0.078 & -0.034 & 0.012 & -0.038 & -0.018 & -0.010 \\ 
			0.6 & -0.062 & -0.013 & -0.005 & -0.025 & 0.042 & 0.004 & -0.065 & -0.010 & 0.009 & -0.090 & -0.015 & -0.044 \\ 
			0.7 & -0.075 & -0.015 & -0.040 & -0.057 & 0.030 & -0.080 & -0.053 & -0.007 & -0.012 & -0.170 & -0.045 & -0.132 \\ 
			0.8 & -0.073 & -0.049 & -0.087 & -0.048 & -0.082 & -0.192 & -0.032 & -0.019 & -0.054 & -0.220 & -0.110 & -0.186 \\ 
			0.9 & -0.084 & -0.078 & -0.139 & -0.050 & -0.147 & -0.245 & -0.048 & -0.047 & -0.120 & -0.219 & -0.133 & -0.232 \\ 
			\hline
			PES & 44.440 & 47.660 & 40.200 & 53.120 & 50.000 & 42.710 & 45.570 & 47.920 & 41.410 & 30.000 & 35.000 & 29.170 \\
			\hline
		\end{tabular}
	}
\end{table}

\begin{figure}[h!]
	\centering
	\includegraphics[scale=0.7]{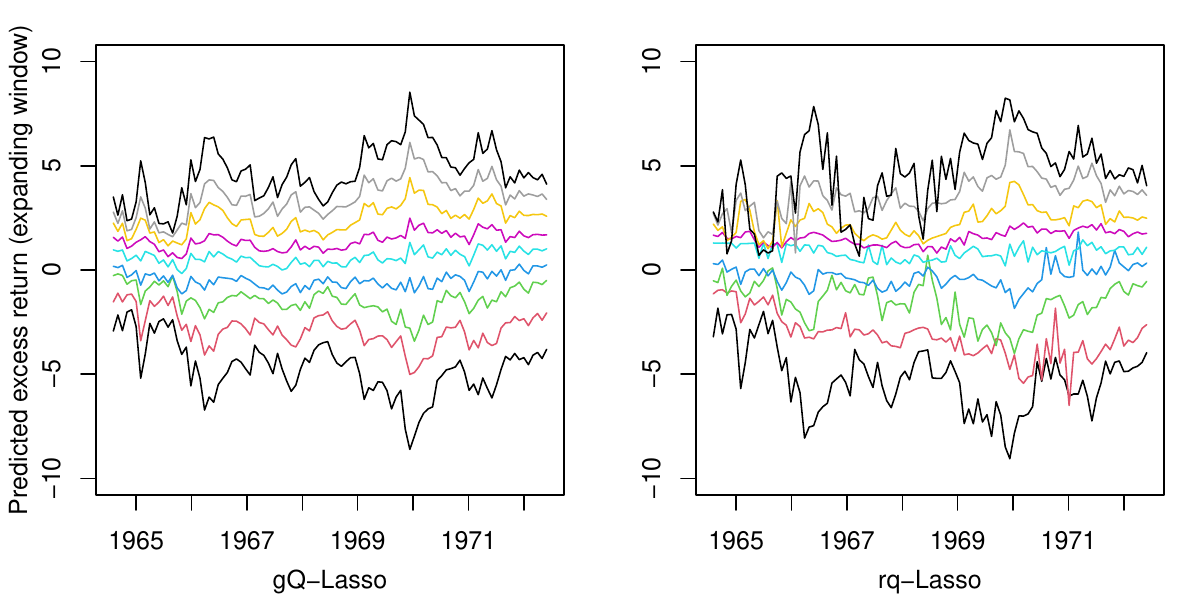}
	\caption{Quantile predictions based on gQ-Lasso and rq-Lasso method} \label{fig:cq}
\end{figure}

Regarding the prediction accuracy, the proposed gQ-Lasso generally outperforms rq-Lasso and GCQR in both $R^2_{mean}$ and $R^2_{\tau}$, while their PES are very similar. 
When compared to the two mean regression, all three quantile methods show noticeable superiority across all evaluation periods and all measures, even including $R^2_{mean}$, which is theoretically advantageous to the mean regression as discussed previously. More importantly, majority of $R^2$'s are positive for the quantile regressions but negative for the mean regressions. Such a finding provides strong empirical evidences that the excess return is hard to predict using traditional mean regression possibly due to heteroscedasticity. Quantile regressions could be better alternatives as suggested by their superior performance.

We provide further evidences in prediction interval and coverage probability to show that the use of mean regression in return prediction can be problematic. In Table \ref{tab:perform3}, we compute the average length and its standard deviation of 80\% and 60\% prediction intervals (IL80 and IL60) and their coverage probabilities (CP80 and CP60), where the interval length is defined as $\hat{Y}_{\tau_k}-\hat{Y}_{\tau_k'}$ for $\tau_k>\tau_k'$, and the coverage probability is computed as  $\frac{1}{N}\sum_{i=1}^{N}I(\hat{Y}_{i,\tau_k'}<Y_i<\hat{Y}_{i,\tau_k})$. We see that the average length of prediction intervals are much smaller for quantile regressions while its standard deviation is much larger, suggesting potential heteroscedasticity. For the coverage probability (CP), gQ-Lasso is the closest to the target interval while the three mean regressions substantially deviates from the target. These findings echo Figure \ref{fig:motivation1} which has motivated our study. 
\begin{table}[h]
	\centering
	\caption{Interval length and coverage probability for out-of-sample prediction.} \label{tab:perform3}
		\footnotesize
		\begin{tabular}{lcccccc}
			\hline
			& IL80 (avg.) & IL80 (sd.) & IL60 (avg.) & IL60 (sd.) & CP80 & CP60 \\ 
			\hline
			gQ-Lasso & 10.280 & 3.297 & 6.830 & 2.295 & 0.809 & 0.616 \\ 
			rq-Lasso & 10.351 & 3.799 & 7.014 & 2.420 & 0.769 & 0.634 \\ 
			GCQR & 10.481 & 3.846 & 7.023 & 2.414 & 0.790 & 0.622 \\
			stepAIC & 14.321 & 0.684 & 9.405 & 0.449 & 0.892 & 0.736 \\
			Lasso & 14.563 & 0.818 & 9.564 & 0.537 & 0.904 & 0.766 \\ 				
			H-Lasso & 14.831 & 0.901 & 9.740 & 0.592 & 0.907 & 0.782 \\ 
			\hline
		\end{tabular}
\end{table}

\section{Conclusion}	

We motivate our work by observing strong evidences of heteroscedasticity in financial and economic data commonly used for equity premium prediction -- an active research area in empirical finance. We advocate for the use of quantile regression and develop a novel penalized quantile regression method that achieves consistent variable selection across multiple quantiles. Our approach provides a comprehensive analysis and produces more reliable predictions over the entire spectrum of quantiles. By ensuring consistent variable selection, our method yields a coherent and highly interpretable model, and effectively mitigates the well-known issue of quantile crossing. We show that under standard quantile regression assumptions, the estimator obtains the $\sqrt{\frac{s \log(p)}{n}}$ convergence rate that is common among lasso-type estimators. 

In our empirical analysis, we demonstrate that the proposed method provides several advantages over  alternative approaches under both mean and quantile regressions frameworks. 
First, compared to the mean regressions, the proposed method delivers superior out-of-sample performance in both $R^2$ and the sign prediction accuracy. Moreover, modeling multiple quantiles through quantile regressions in general can provide a full spectrum of the distribution estimation, which in turn allows practitioners and researchers to conduct more comprehensive examinations. As demonstrated by our empirical results, we found that several important variables have changed the sign of coefficients from lower to upper quantiles due to heteroscedastic relationship, which otherwise cannot be unveiled by any mean regressions.
Second, for variable selection, the proposed method demonstrates a clear advantage over two popular Lasso-based quantile regressions, rq-Lasso and GCQR; both selects different variables at different quantiles, creating unnecessary complexity for interpretation. 
Third, we show that the proposed approach significantly reduces the issue of crossing quantile predictions when multiple quantiles are modeled, compared to the two alternative quantile methods. This potentially produces more reliable model estimation and predictions, which have been evidenced by both simulation and the empirical study of the equity premium prediction. 
With these advantages, the quantile predictions using the proposed method outperforms alternative methods for multiple periods of out-of-sample evaluation. Evidence of return predictability is shown under quantile regression approaches when major recession periods are excluded, while the conclusions are mixed in existing studies \citep{campbell2008predicting,welch2008comprehensive,goyal2024comprehensive}, which rely on the OLS framework.
While we showcase our method through an important empirical finance problem, assessing equity premium predictability, our method can be broadly adopted to other fields and applications where the interests are modeling multiple quantiles with a large set of covariates.

Although the proposed method significantly mitigates the issue of quantile crossing, it does not completely eliminate the problem. Alternative approaches, such as those by \citet{bondell2010noncrossing} and \citet{cqQu}, incorporate constraints into the optimization process to strictly prevent crossing in the fitted values; however, such constraints have not yet been explored within the framework of penalized quantile regression.
Nevertheless, in our empirical analysis, we implement the constrained quantile regression (constrQR) method proposed by \citet{bondell2010noncrossing} to the equity premium prediction problem. Table S6 in the supplemental materials compares its out-of-sample performance with the proposed gQ-Lasso. Although constrQR successfully prevents quantile crossing, it underperforms the gQ-Lasso approach across all other out-of-sample prediction metrics. Notably, constrQR yields negative $R^2_{mean}$ values for all evaluation periods, indicating that predictions from constrQR are no better than the historical average.
A notable advantage of these constraints is that the problem remains a linear programming problem, which can be solved using standard quantile regression algorithms. Nevertheless, even without constraints, these algorithms have been demonstrated to be computationally slow for penalized quantile regression \citep{yi2017semismooth, tanSmooth}. Introducing constraints would further exacerbate this computational inefficiency. Thus, developing a fast algorithm to integrate these constraints into penalized quantile regression represents a promising direction for future research.

\bibliographystyle{apalike}
\bibliography{ref}

\end{document}


\def\spacingset#1{\renewcommand{\baselinestretch}%
{#1}\small\normalsize} \spacingset{1}

  \title{\bf Supplemental Material to Quantile Predictions for Equity Premium using Penalized Quantile Regression with Consistent Variable Selection across Multiple Quantiles}
  \author{Shaobo Li\\
    School of Business, University of Kansas\\
    and \\
    Ben Sherwood\\
    School of Business, University of Kansas}
  \maketitle

\section{Additional results}

\subsection{Additional results from simulation study}
Table \ref{tab:simu1} and \ref{tab:simu2} present the false positive and false negative rates from the simulations. From Table \ref{tab:simu1}, we see that in general, gQ-Lasso tends to select a larger models than rq-Lasso and GCQR, where the latter has a much lower FP. This is not surprising as gQ-Lasso may pick false signals at certain quantiles then keep them across all quantiles. On the other hand, rq-Lasso and GCQR do not guarantee the same sparsity across all quantiles, so that potential false signals may only appear at certain quantiles. Furthermore, the low FP of GCQR is likely due to its adaptive-Lasso penalty, which satisfies the oracle property. We also observe that compared to homoscedastic error, gQ-Lasso and rq-Lasso have smaller FP under the settings with heteroscedastic error, but an opposite pattern is observed for GCQR. As for FN, Table \ref{tab:simu2} shows that gQ-Lasso substantially outperforms the other two methods especially for settings with complex errors. For the normal error, none of the three methods miss correct signals over the entire 100 replicates. 

Table \ref{tab:comparison-ME} and \ref{tab:comparison-NCQ} report model error and number of crossing-quantiles for comparison between the proposed method and mrq-$L_\infty$. Results are based on 10 replicates due to extreme computational time of mrq-$L_\infty$ method. The results show that mrq-$L_\infty$ slightly outperforms gQ-Lasso under $p=20$, but substantially worse for $p=100$.
	
\begin{table}[h!]
	\centering
	\caption{Average false positive (standard error) for variable selection} \label{tab:simu1}
	\resizebox{\textwidth}{!}{%
		\begin{tabular}{llccccccccccc}
			\hline
			$\epsilon$ & $N$ &  \multicolumn{3}{c}{p=20} && \multicolumn{3}{c}{p=50} && \multicolumn{3}{c}{p=100} \\
			\cline{3-5}  \cline{7-9} \cline{11-13}  
			& & gQ-Lasso & rq-Lasso & GCQR && gQ-Lasso & rq-Lasso & GCQR && gQ-Lasso & rq-Lasso & GCQR\\
			\hline
			N(0,1) & 200 & 8.730 & 6.754 & 1.306 && 13.360 & 10.552 & 2.364 && 15.840 & 14.048 & 3.308 \\ 
			& & \textit{(3.440)} & \textit{(2.067)} & \textit{(1.049)} && \textit{(5.885)} & \textit{(3.598)} & \textit{(1.318)} && \textit{(8.147)} & \textit{(5.677)} & \textit{(2.058)} \\ 
			& 500 & 8.860 & 6.148 & 0.330 && 12.630 & 9.402 & 0.408 && 17.110 & 13.052 & 0.520 \\ 
			& & \textit{(3.490)} & \textit{(1.784)} & \textit{(0.438)} && \textit{(5.863)} & \textit{(3.502)} & \textit{(0.549)} && \textit{(9.056)} & \textit{(4.888)} & \textit{(0.613)} \\ 
			& 1000 & 8.930 & 5.958 & 0.034 && 13.100 & 9.442 & 0.018 && 18.210 & 12.798 & 0.030 \\ 
			& & \textit{(3.400)} & \textit{(1.986)} & \textit{(0.095)} && \textit{(5.342)} & \textit{(3.327)} & \textit{(0.064)} && \textit{(9.524)} & \textit{(3.662)} & \textit{(0.096)} \\ 
			\hline
			$T_2$ &  200 & 8.650 & 6.568 & 2.390 && 11.940 & 9.942 & 4.406 && 15.120 & 13.804 & 5.930 \\ 
			& & \textit{(3.304)} & \textit{(2.095)} & \textit{(1.460)} && \textit{(5.750)} & \textit{(3.027)} & \textit{(2.155)} && \textit{(8.086)} & \textit{(5.098)} & \textit{(2.795)} \\ 
			& 500 & 9.180 & 6.362 & 1.016 && 13.780 & 9.824 & 1.598 && 16.050 & 12.268 & 1.940 \\ 
			& & \textit{(3.705)} & \textit{(1.955)} & \textit{(0.716)} && \textit{(6.458)} & \textit{(3.346)} & \textit{(1.119)} && \textit{(8.281)} & \textit{(4.348)} & \textit{(1.390)} \\ 
			& 1000 & 9.220 & 6.132 & 0.214 && 14.020 & 9.834 & 0.396 && 16.470 & 11.928 & 0.422 \\ 
			& & \textit{(3.386)} & \textit{(1.722)} & \textit{(0.310)} && \textit{(6.627)} & \textit{(3.377)} & \textit{(0.514)} && \textit{(8.241)} & \textit{(4.027)} & \textit{(0.479)} \\ 
			\hline
		  heter. & 200 & 5.590 & 4.224 & 3.078 && 6.700 & 5.450 & 5.562 && 6.400 & 4.834 & 6.882 \\ 
			& & \textit{(3.838)} & \textit{(2.419)} & \textit{(1.776)} && \textit{(5.608)} & \textit{(3.458)} & \textit{(2.912)} && \textit{(5.735)} & \textit{(3.373)} & \textit{(4.093)} \\ 
			& 500 & 7.980 & 5.964 & 3.226 && 10.480 & 7.906 & 5.424 && 13.370 & 9.272 & 8.996 \\ 
			& & \textit{(3.687)} & \textit{(2.058)} & \textit{(1.771)} && \textit{(5.654)} & \textit{(3.360)} & \textit{(2.392)} && \textit{(7.396)} & \textit{(4.275)} & \textit{(3.858)} \\ 
			& 1000 & 8.620 & 5.764 & 2.358 && 11.640 & 8.752 & 4.770 && 14.490 & 11.600 & 8.062 \\ 
			& & \textit{(3.487)} & \textit{(1.939)} & \textit{(1.445)} && \textit{(6.028)} & \textit{(3.083)} & \textit{(2.433)} && \textit{(6.722)} & \textit{(4.098)} & \textit{(3.533)} \\  
			\hline
			asym. &   200 & 7.170 & 6.230 & 3.030 && 11.080 & 10.056 & 5.104 && 13.790 & 13.446 & 8.360 \\ 
			&& \textit{(3.414)} & \textit{(2.060)} & \textit{(1.734)} && \textit{(6.147)} & \textit{(3.845)} & \textit{(2.897)} && \textit{(7.263)} & \textit{(5.854)} & \textit{(4.208)} \\ 
			&500 & 8.260 & 6.352 & 1.872 && 12.870 & 10.338 & 2.914 && 15.400 & 13.212 & 4.304 \\ 
			&& \textit{(3.529)} & \textit{(1.964)} & \textit{(1.393)} && \textit{(7.894)} & \textit{(4.061)} & \textit{(1.885)} && \textit{(8.962)} & \textit{(4.803)} & \textit{(2.669)} \\ 
			&1000 & 8.030 & 6.038 & 0.786 && 12.360 & 9.734 & 1.032 && 14.610 & 11.612 & 1.310 \\ 
			&& \textit{(3.597)} & \textit{(1.912)} & \textit{(0.747)} && \textit{(6.994)} & \textit{(3.750)} & \textit{(0.908)} && \textit{(7.256)} & \textit{(4.176)} & \textit{(0.939)} \\ 
			\hline
			heter. + asym. & 200 & 4.250 & 3.314 & 2.416 && 4.860 & 4.006 & 4.800 && 3.580 & 3.496 & 5.269 \\ 
			&& \textit{(3.514)} & \textit{(1.909)} & \textit{(1.529)} && \textit{(5.472)} & \textit{(2.953)} & \textit{(3.154)} && \textit{(5.101)} & \textit{(2.915)} & \textit{(3.631)} \\ 
			& 500 & 5.850 & 3.782 & 2.698 && 8.500 & 6.206 & 4.538 && 8.680 & 6.384 & 6.474 \\ 
			&&  \textit{(3.380)} & \textit{(1.692)} & \textit{(1.670)} && \textit{(6.414)} & \textit{(3.516)} & \textit{(2.797)} && \textit{(6.903)} & \textit{(4.122)} & \textit{(4.044)} \\ 
			& 1000 & 7.170 & 4.508 & 2.294 && 9.940 & 6.772 & 4.220 && 11.830 & 7.384 & 5.274 \\ 
			&&  \textit{(3.447)} & \textit{(1.911)} & \textit{(1.311)} && \textit{(5.517)} & \textit{(2.703)} & \textit{(2.203)} && \textit{(6.784)} & \textit{(3.594)} & \textit{(3.317)} \\ 
			\hline
		\end{tabular}
	}
\end{table}

\begin{table}[h!]
	\centering
	\caption{Average false negative (standard error) for variable selection} \label{tab:simu2}
	\resizebox{\textwidth}{!}{%
		\begin{tabular}{lccccccccccc}
			\hline
			$N$ &  \multicolumn{3}{c}{p=20} && \multicolumn{3}{c}{p=50} && \multicolumn{3}{c}{p=100} \\
			\cline{2-4}  \cline{6-8} \cline{10-12}  
			& gQ-Lasso & rq-Lasso & GCQR && gQ-Lasso & rq-Lasso & GCQR && gQ-Lasso & rq-Lasso & GCQR\\
			\hline
			\multicolumn{3}{l}{Error: $\epsilon\sim N(0,1)$}&&&&&&&&&\\
			\cline{1-2}
			200 & 0.000 & 0.000 & 0.000 && 0.000 & 0.000 & 0.000 && 0.000 & 0.000 & 0.000 \\ 
			& \textit{(0.000)} & \textit{(0.000)} & \textit{(0.000)} && \textit{(0.000)} & \textit{(0.000)} & \textit{(0.000)} && \textit{(0.000)} & \textit{(0.000)} & \textit{(0.000)} \\ 
			500 & 0.000 & 0.000 & 0.000 && 0.000 & 0.000 & 0.000 && 0.000 & 0.000 & 0.000 \\ 
			& \textit{(0.000)} & \textit{(0.000)} & \textit{(0.000)} && \textit{(0.000)} & \textit{(0.000)} & \textit{(0.000)} && \textit{(0.000)} & \textit{(0.000)} & \textit{(0.000)} \\ 
			1000 & 0.000 & 0.000 & 0.000 && 0.000 & 0.000 & 0.000 && 0.000 & 0.000 & 0.000 \\ 
			& \textit{(0.000)} & \textit{(0.000)} & \textit{(0.000)} && \textit{(0.000)} & \textit{(0.000)} & \textit{(0.000)} && \textit{(0.000)} & \textit{(0.000)} & \textit{(0.000)} \\ 			
			\hline
			\multicolumn{3}{l}{Error: $\epsilon\sim t(2)$}&&&&&&&&&\\
			\cline{1-2}
			200 & 0.000 & 0.030 & 0.040 && 0.000 & 0.064 & 0.040 && 0.000 & 0.076 & 0.038 \\ 
			& \textit{(0.000)} & \textit{(0.072)} & \textit{(0.090)} && \textit{(0.000)} & \textit{(0.124)} & \textit{(0.085)} && \textit{(0.000)} & \textit{(0.123)} & \textit{(0.079)} \\ 
			500 & 0.000 & 0.006 & 0.006 && 0.000 & 0.000 & 0.000 && 0.000 & 0.004 & 0.002 \\ 
			& \textit{(0.000)} & \textit{(0.034)} & \textit{(0.034)} && \textit{(0.000)} & \textit{(0.000)} & \textit{(0.000)} && \textit{(0.000)} & \textit{(0.028)} & \textit{(0.020)} \\ 
			1000 & 0.000 & 0.000 & 0.000 && 0.000 & 0.000 & 0.000 && 0.000 & 0.000 & 0.000 \\ 
			& \textit{(0.000)} & \textit{(0.000)} & \textit{(0.000)} && \textit{(0.000)} & \textit{(0.000)} & \textit{(0.000)} && \textit{(0.000)} & \textit{(0.000)} & \textit{(0.000)} \\  			
			\hline
			\multicolumn{3}{l}{Error: heteroskedasticity}&&&&&&&&&\\
			\cline{1-2}
			200 & 0.500 & 1.616 & 1.154 && 0.810 & 2.066 & 1.158 && 1.620 & 2.674 & 1.640 \\ 
			& \textit{(0.835)} & \textit{(0.783)} & \textit{(0.667)} && \textit{(1.022)} & \textit{(0.724)} & \textit{(0.783)} && \textit{(1.179)} & \textit{(0.601)} & \textit{(0.869)} \\ 
			500 & 0.000 & 0.486 & 0.188 && 0.020 & 0.830 & 0.242 && 0.040 & 1.094 & 0.256 \\ 
			& \textit{(0.000)} & \textit{(0.283)} & \textit{(0.220)} && \textit{(0.141)} & \textit{(0.428)} & \textit{(0.238)} && \textit{(0.197)} & \textit{(0.460)} & \textit{(0.318)} \\ 
			1000 & 0.000 & 0.260 & 0.096 && 0.000 & 0.404 & 0.114 && 0.010 & 0.490 & 0.094 \\ 
			& \textit{(0.000)} & \textit{(0.149)} & \textit{(0.115)} && \textit{(0.000)} & \textit{(0.175)} & \textit{(0.137)} && \textit{(0.100)} & \textit{(0.189)} & \textit{(0.146)} \\ 
			\hline
			\multicolumn{3}{l}{Error: asymmetric}&&&&&&&&&\\
			\cline{1-2}
			200 & 0.000 & 0.130 & 0.096 && 0.000 & 0.202 & 0.056 && 0.000 & 0.330 & 0.084 \\ 
			& \textit{(0.000)} & \textit{(0.194)} & \textit{(0.143)} && \textit{(0.000)} & \textit{(0.227)} & \textit{(0.121)} && \textit{(0.000)} & \textit{(0.281)} & \textit{(0.137)} \\ 
			500 & 0.000 & 0.010 & 0.012 && 0.000 & 0.014 & 0.010 && 0.000 & 0.026 & 0.004 \\ 
			& \textit{(0.000)} & \textit{(0.044)} & \textit{(0.048)} && \textit{(0.000)} & \textit{(0.051)} & \textit{(0.044)} && \textit{(0.000)} & \textit{(0.073)} & \textit{(0.028)} \\ 
			1000 & 0.000 & 0.000 & 0.002 && 0.000 & 0.008 & 0.004 && 0.000 & 0.004 & 0.002 \\ 
			& \textit{(0.000)} & \textit{(0.000)} & \textit{(0.020)} && \textit{(0.000)} & \textit{(0.039)} & \textit{(0.028)} && \textit{(0.000)} & \textit{(0.028)} & \textit{(0.020)} \\ 
			\hline
			\multicolumn{3}{l}{Error: heter. + asym.}&&&&&&&&&\\
			\cline{1-2}
			200 & 1.180 & 2.250 & 1.770 && 1.950 & 2.716 & 2.014 && 2.440 & 2.960 & 2.176 \\ 
			& \textit{(0.989)} & \textit{(0.654)} & \textit{(0.670)} && \textit{(1.067)} & \textit{(0.506)} & \textit{(0.630)} && \textit{(0.903)} & \textit{(0.399)} & \textit{(0.663)} \\ 
			500 & 0.220 & 1.520 & 1.036 && 0.540 & 1.882 & 1.010 && 0.780 & 2.086 & 0.904 \\ 
			& \textit{(0.462)} & \textit{(0.520)} & \textit{(0.546)} && \textit{(0.784)} & \textit{(0.511)} & \textit{(0.686)} && \textit{(0.949)} & \textit{(0.468)} & \textit{(0.652)} \\ 
			1000 & 0.010 & 0.952 & 0.490 && 0.070 & 1.202 & 0.514 && 0.080 & 1.516 & 0.456 \\ 
			& \textit{(0.100)} & \textit{(0.449)} & \textit{(0.431)} && \textit{(0.293)} & \textit{(0.453)} & \textit{(0.405)} && \textit{(0.273)} & \textit{(0.410)} & \textit{(0.413)} \\ 	
			\hline
		\end{tabular}
	}
\end{table}

\begin{table}[htbp]
	\centering
	\caption{Model error (ME) of gQ-Lasso and mrq-$L_\infty$ based on 10 replicates.}
	\label{tab:comparison-ME}
	\resizebox{\textwidth}{!}{%
		\begin{tabular}{llcccccccccc}
			\toprule
			& & \multicolumn{2}{c}{\( N(0,1) \)} & \multicolumn{2}{c}{\( t(2) \)} & \multicolumn{2}{c}{Heter.} & \multicolumn{2}{c}{Asym.} & \multicolumn{2}{c}{Heter.+Asym.} \\
			\cmidrule(lr){3-4} \cmidrule(lr){5-6} \cmidrule(lr){7-8} \cmidrule(lr){9-10} \cmidrule(lr){11-12}
			$N$ & $p$ & {gQ-Lasso} & {mrq-$L_\infty$} & {gQ-Lasso} & {mrq-$L_\infty$} & {gQ-Lasso} & {mrq-$L_\infty$} & {gQ-Lasso} & {mrq-$L_\infty$} & {gQ-Lasso} & {mrq-$L_\infty$} \\
			\midrule 
			200 & 20 & 0.122 & 0.085 & 0.597 & 0.396 & 5.564 & 5.740 & 0.607 & 0.360 & 22.476 & 28.636 \\
			& & \textit{(0.049)} & \textit{(0.026)} & \textit{(0.164)} & \textit{(0.095)} & \textit{(1.383)} & \textit{(1.447)} & \textit{(0.155)} & \textit{(0.127)} & \textit{(7.527)} & \textit{(10.698)} \\
			& 100 & 0.247 & 0.653 & 0.848 & 1.658 & 8.578 & 37.820 & 0.845 & 3.249 & 26.904 & 161.467 \\
			& & \textit{(0.078)} & \textit{(0.096)} & \textit{(0.422)} & \textit{(0.478)} & \textit{(1.601)} & \textit{(4.884)} & \textit{(0.218)} & \textit{(0.713)} & \textit{(9.034)} & \textit{(33.371)} \\
			\midrule 
			500 & 20 & 0.041 & 0.033 & 0.276 & 0.191 & 2.197 & 2.199 & 0.242 & 0.140 & 10.737 & 9.733 \\
			& & \textit{(0.005)} & \textit{(0.007)} & \textit{(0.063)} & \textit{(0.040)} & \textit{(0.834)} & \textit{(0.630)} & \textit{(0.078)} & \textit{(0.050)} & \textit{(1.928)} & \textit{(2.618)} \\
			& 100 & 0.080 & 0.327 & 0.512 & 1.002 & 3.664 & 17.363 & 0.435 & 1.774 & 13.893 & 96.919 \\
			& & \textit{(0.032)} & \textit{(0.061)} & \textit{(0.118)} & \textit{(0.197)} & \textit{(1.210)} & \textit{(2.225)} & \textit{(0.149)} & \textit{(0.230)} & \textit{(2.470)} & \textit{(15.578)} \\
			\bottomrule
		\end{tabular}
	}
\end{table}

\begin{table}[htbp]
	\centering
	\caption{Number of crossing quantiles (NCQ) of gQ-Lasso and mrq-$L_\infty$ based on 10 replicates.}
	\label{tab:comparison-NCQ}
	\resizebox{\textwidth}{!}{%
		\begin{tabular}{llcccccccccc}
			\toprule
			& & \multicolumn{2}{c}{\( N(0,1) \)} & \multicolumn{2}{c}{\( t(2) \)} & \multicolumn{2}{c}{Heter.} & \multicolumn{2}{c}{Asym.} & \multicolumn{2}{c}{Heter.+Asym.} \\
			\cmidrule(lr){3-4} \cmidrule(lr){5-6} \cmidrule(lr){7-8} \cmidrule(lr){9-10} \cmidrule(lr){11-12}
			$N$ & $p$ & {gQ-Lasso} & {mrq-$L_\infty$} & {gQ-Lasso} & {mrq-$L_\infty$} & {gQ-Lasso} & {mrq-$L_\infty$} & {gQ-Lasso} & {mrq-$L_\infty$} & {gQ-Lasso} & {mrq-$L_\infty$} \\
			\midrule
			200 & 20 & 1.4 & 0.6 & 2.7 & 0.5 & 0.2 & 1.4 & 0.6 & 0.1 & 0.3 & 5.4 \\
			& & \textit{(1.5)} & \textit{(1.1)} & \textit{(2.3)} & \textit{(1.3)} & \textit{(0.6)} & \textit{(2.2)} & \textit{(0.8)} & \textit{(0.3)} & \textit{(0.9)} & \textit{(4.8)} \\
			& 100 & 1.5 & 102.7 & 2.8 & 107.0 & 0.0 & 101.8 & 0.7 & 103.3 & 0.0 & 96.4 \\
			& & \textit{(2.7)} & \textit{(12.0)} & \textit{(1.8)} & \textit{(8.9)} & \textit{(0.0)} & \textit{(13.6)} & \textit{(0.9)} & \textit{(17.8)} & \textit{(0.0)} & \textit{(9.7)} \\
			\midrule
			500 & 20 & 0.0 & 0.0 & 0.5 & 0.0 & 0.0 & 0.1 & 0.0 & 0.1 & 0.0 & 0.3 \\
			& & \textit{(0.0)} & \textit{(0.0)} & \textit{(1.0)} & \textit{(0.0)} & \textit{(0.0)} & \textit{(0.3)} & \textit{(0.0)} & \textit{(0.3)} & \textit{(0.0)} & \textit{(0.7)} \\
			& 100 & 0.7 & 143.1 & 0.8 & 157.2 & 0.0 & 139.1 & 0.2 & 162.4 & 0.0 & 167.0 \\
			& & \textit{(1.3)} & \textit{(24.2)} & \textit{(0.8)} & \textit{(11.9)} & \textit{(0.0)} & \textit{(22.5)} & \textit{(0.6)} & \textit{(13.3)} & \textit{(0.0)} & \textit{(16.1)} \\
			\bottomrule
		\end{tabular}
	}
\end{table}
\subsection{Additional results from empirical analysis}

To take into account selection uncertainty, we conduct 1000 bootstrap sampling and compute the percentage of selection for each variable across each method. The result is reported in Table \ref{tab:realdatacoef}. These results, which indicate the strength of each variable.	
	The bootstrap results also reflect model size. Compared to the three mean regressions, gQ-Lasso on average picks the largest model with almost every predictor exceeding 50\% selection rate. Although we do not know the true model, all the examined predictors have been documented and shown evidence of predictivity for stock return. As for the rq-Lasso and GCQR, it is not straightforward to evaluate the model size, thus not comparable to other methods. 

\begin{sidewaystable}[h]
	\centering
	\caption{\footnotesize Comparison of variable selection frequency (\%) based on 1000 bootstrap for different modeling approaches based on the full sample period of 1926 -- 2021.} \label{tab:realdatacoef}
	\resizebox{\textwidth}{!}{%
		\begin{tabular}{llccccccccccccccccccccc}
			\hline
			Variable & Proposed & \multicolumn{9}{c}{rq-Lasso (for each $\tau$)} & \multicolumn{9}{c}{GCQR (for each $\tau$)} & stepAIC & Lasso & H-Lasso\\
			\cline{3-11} \cline{12-20}
			&& 0.1 & 0.2 & 0.3 & 0.4 & 0.5 & 0.6 & 0.7 & 0.8 & 0.9 & 0.1 & 0.2 & 0.3 & 0.4 & 0.5 & 0.6 & 0.7 & 0.8 & 0.9 &&&\\
			\hline
			avgcor & 86.7 & 56.3 & 57.1 & 81.4 & 97.6 & 99.6 & 96.4 & 81.9 & 68.8 & 83.4 & 54.9 & 58.5 & 80.8 & 94.1 & 96.9 & 87.9 & 72.0 & 56.2 & 99.9 & 61.3 & 83.6 & 96.8 \\ 
			b.m & 53.6 & 50.2 & 43.4 & 38.6 & 35.7 & 44.4 & 45.6 & 52.2 & 72.5 & 50.7 & 47.8 & 48.8 & 44.0 & 46.0 & 51.0 & 50.9 & 53.8 & 62.5 & 77.9 & 90.0 & 87.8 & 31.5 \\ 
			d.e & 56.6 & 71.3 & 61.5 & 46.7 & 29.7 & 37.0 & 44.9 & 55.3 & 39.9 & 48.3 & 50.1 & 48.9 & 41.2 & 29.0 & 29.2 & 28.8 & 34.9 & 26.3 & 95.1 & 1.2 & 55.2 & 32.2 \\ 
			d.p & 5.9 & 13.8 & 11.3 & 13.6 & 28.0 & 41.5 & 46.3 & 57.7 & 44.5 & 12.5 & 20.1 & 18.0 & 16.9 & 23.5 & 29.9 & 29.8 & 35.6 & 30.1 & 50.0 & 67.8 & 31.9 & 22.4 \\ 
			d.y & 40.0 & 40.8 & 31.6 & 26.3 & 27.6 & 21.6 & 18.2 & 16.6 & 24.8 & 41.6 & 40.6 & 38.8 & 34.9 & 36.1 & 34.2 & 31.5 & 34.5 & 33.6 & 55.6 & 65.4 & 33.2 & 23.7 \\ 
			dfr & 92.1 & 82.9 & 92.5 & 88.7 & 80.0 & 71.1 & 54.6 & 60.0 & 54.9 & 83.5 & 72.2 & 83.4 & 82.5 & 75.8 & 68.4 & 54.3 & 54.3 & 51.2 & 98.1 & 67.2 & 81.7 & 78.8 \\ 
			dfy & 90.9 & 81.6 & 83.0 & 80.7 & 59.5 & 48.7 & 41.8 & 57.6 & 63.3 & 77.5 & 71.7 & 75.2 & 76.3 & 63.7 & 57.6 & 43.6 & 48.1 & 55.5 & 98.0 & 82.2 & 73.3 & 59.9 \\ 
			dtoat & 78.0 & 67.7 & 52.6 & 45.4 & 46.1 & 47.1 & 47.8 & 66.5 & 52.0 & 66.4 & 52.6 & 46.9 & 41.8 & 43.9 & 46.0 & 40.9 & 50.7 & 40.9 & 93.1 & 67.6 & 61.5 & 40.0 \\ 
			dtoy & 98.4 & 93.1 & 87.7 & 76.3 & 53.7 & 55.1 & 78.4 & 80.9 & 89.4 & 97.7 & 90.5 & 89.8 & 81.9 & 70.9 & 72.5 & 86.1 & 90.0 & 93.0 & 100.0 & 61.4 & 66.9 & 52.6 \\ 
			e.p & 74.6 & 55.2 & 75.8 & 76.4 & 57.3 & 38.8 & 26.9 & 22.4 & 27.6 & 54.0 & 44.8 & 55.4 & 54.0 & 42.1 & 36.4 & 28.9 & 26.5 & 26.2 & 96.0 & 33.6 & 45.1 & 58.2 \\ 
			fbm & 95.2 & 98.1 & 98.7 & 98.1 & 97.3 & 98.0 & 98.6 & 95.0 & 71.4 & 95.3 & 97.1 & 98.5 & 98.4 & 98.5 & 99.4 & 99.0 & 97.7 & 86.9 & 100.0 & 99.9 & 98.2 & 97.4 \\ 
			infl & 78.7 & 81.9 & 96.6 & 97.9 & 96.5 & 94.6 & 81.9 & 78.4 & 80.4 & 66.1 & 72.5 & 91.2 & 94.5 & 92.5 & 89.7 & 74.9 & 75.3 & 76.3 & 100.0 & 66.7 & 86.1 & 96.1 \\ 
			ltr & 83.7 & 83.9 & 82.3 & 68.2 & 70.9 & 90.9 & 83.2 & 94.8 & 94.6 & 69.7 & 71.1 & 69.3 & 57.2 & 61.7 & 82.0 & 73.7 & 85.4 & 83.2 & 99.1 & 61.5 & 84.0 & 66.7 \\ 
			lty & 52.3 & 21.0 & 15.8 & 25.4 & 31.1 & 35.6 & 38.9 & 44.1 & 37.1 & 30.5 & 15.5 & 16.4 & 18.0 & 21.3 & 24.6 & 22.0 & 24.3 & 16.2 & 69.4 & 48.2 & 13.9 & 28.6 \\ 
			lzrt & 64.4 & 59.9 & 71.4 & 76.5 & 74.5 & 72.5 & 57.6 & 61.8 & 67.5 & 33.4 & 35.3 & 49.5 & 56.8 & 59.7 & 56.8 & 40.8 & 43.0 & 46.7 & 97.9 & 18.3 & 57.4 & 72.8 \\ 
			ntis & 85.7 & 89.8 & 90.1 & 84.6 & 75.7 & 56.8 & 48.8 & 57.3 & 58.3 & 73.0 & 76.6 & 78.4 & 75.2 & 68.9 & 58.0 & 48.2 & 58.2 & 57.9 & 99.0 & 84.7 & 90.4 & 75.4 \\ 
			ogap & 70.5 & 60.6 & 62.0 & 66.0 & 78.7 & 85.6 & 64.5 & 63.7 & 68.3 & 44.6 & 39.7 & 41.3 & 44.5 & 60.9 & 67.6 & 47.2 & 45.7 & 49.7 & 97.2 & 29.1 & 73.2 & 76.8 \\ 
			rdsp & 61.2 & 57.9 & 60.9 & 71.2 & 63.7 & 53.7 & 44.3 & 53.8 & 85.2 & 48.9 & 44.5 & 49.9 & 57.5 & 55.7 & 50.6 & 38.9 & 43.2 & 64.0 & 96.3 & 45.1 & 63.1 & 62.4 \\ 
			skvw & 72.8 & 71.3 & 79.4 & 86.9 & 85.3 & 92.3 & 93.7 & 93.3 & 68.3 & 53.6 & 55.0 & 69.2 & 79.1 & 75.0 & 82.6 & 81.5 & 79.2 & 56.1 & 100.0 & 46.9 & 73.5 & 83.5 \\ 
			svar & 100.0 & 94.2 & 68.3 & 47.6 & 41.3 & 43.0 & 65.2 & 99.1 & 100.0 & 100.0 & 96.4 & 80.3 & 64.8 & 64.8 & 70.3 & 82.4 & 99.0 & 99.8 & 100.0 & 30.1 & 50.2 & 38.1 \\ 
			tail & 76.9 & 65.1 & 63.7 & 65.5 & 60.1 & 57.7 & 55.3 & 63.8 & 80.7 & 45.8 & 39.6 & 36.8 & 38.5 & 34.9 & 34.5 & 30.9 & 37.3 & 55.3 & 95.8 & 37.6 & 62.0 & 55.2 \\ 
			tbl & 18.4 & 48.4 & 66.1 & 59.3 & 58.7 & 56.9 & 32.9 & 33.3 & 21.0 & 12.7 & 25.8 & 36.5 & 32.9 & 33.4 & 34.0 & 24.3 & 23.9 & 12.7 & 77.8 & 80.1 & 63.5 & 61.9 \\ 
			tms & 97.5 & 90.0 & 81.3 & 68.1 & 55.6 & 50.1 & 40.4 & 43.0 & 48.0 & 88.1 & 82.1 & 75.1 & 69.2 & 64.2 & 58.6 & 46.0 & 45.9 & 42.4 & 99.0 & 1.5 & 70.0 & 57.0 \\ 
			wtexas & 77.5 & 67.0 & 66.5 & 69.7 & 90.2 & 94.9 & 88.9 & 93.4 & 87.4 & 52.5 & 44.2 & 46.8 & 49.8 & 74.2 & 82.7 & 68.2 & 74.5 & 65.1 & 99.9 & 42.3 & 77.5 & 87.9 \\ 
			\hline
		\end{tabular}
	}
\end{sidewaystable}

To prevent crossing quantiles, \citet{bondell2010noncrossing} proposed a constrained quantile regression (constrQR) under fixed-dimension setup. Table \ref{tab:constrQR} compares gQ-Lasso and constrQR on out-of-sample performance for the same evaluation periods as in Table 4 in the paper. We see that, as expected, PCQ is 0 for constrQR throughout, but PES, $R^2_{mean}$ and $R^2_{\tau}$ are all in favor of the proposed gQ-Lasso.
\begin{table}[h]
\centering
\caption{Comparison between gQ-Lasso and constrQR on out-of-sample performance.} \label{tab:constrQR}
\resizebox{\textwidth}{!}{%
\begin{tabular}{lcccccccc}
\hline
&\multicolumn{2}{c}{1965-2021} & \multicolumn{2}{c}{1965-1972} & \multicolumn{2}{c}{1976-2007} & \multicolumn{2}{c}{2010-2019}\\
\cline{2-3} \cline{4-5} \cline{6-7} \cline{8-9}
& gQ-Lasso & constrQR & gQ-Lasso & constrQR & gQ-Lasso & constrQR & gQ-Lasso & constrQR \\ 
\hline			
$R^2_{mean}$ & -0.003 & -0.041 & 0.038 & -0.012 & 0.034 & -0.010 & 0.008 & -0.010 \\ 
\hline
\multicolumn{2}{l}{$R^2_{\tau}$} &&&&&&&\\
\cline{1-2}
$\tau=0.1$      & 0.049 & 0.041 & 0.033 & 0.067 & 0.056 & 0.039 & 0.057 & 0.042 \\ 
0.2      & 0.008 & -0.006 & 0.018 & 0.006 & 0.016 & -0.005 & 0.025 & 0.015 \\ 
0.3      & -0.005 & -0.010 & 0.007 & 0.014 & 0.017 & -0.001 & -0.016 & -0.017 \\ 
0.4      & -0.002 & -0.010 & 0.014 & -0.005 & 0.016 & -0.001 & -0.026 & -0.016 \\ 
0.5      & 0.002 & -0.016 & 0.042 & 0.003 & 0.010 & -0.014 & -0.013 & -0.019 \\ 
0.6      & 0.016 & -0.012 & 0.041 & -0.003 & 0.026 & -0.012 & -0.003 & -0.009 \\ 
0.7      & 0.023 & -0.004 & 0.044 & -0.018 & 0.027 & -0.005 & 0.025 & 0.016 \\ 
0.8      & 0.059 & 0.025 & 0.131 & 0.058 & 0.040 & 0.005 & 0.083 & 0.089 \\ 
0.9      & 0.101 & 0.081 & 0.218 & 0.214 & 0.060 & 0.012 & 0.185 & 0.176 \\ 
\hline 
PES & 41.520 & 44.150 & 44.790 & 45.830 & 41.930 & 47.400 & 29.170 & 29.170 \\ 
PCQ & 1.900 & 0.000 & 0.000 & 0.000 & 2.080 & 0.000 & 0.830 & 0.000 \\ 
\hline 
\end{tabular}
}
\end{table}

\clearpage

\newpage

\section{Proofs}

\subsection{Proof of Lemma 1}

\begin{proof}
	By definition,
	\begin{equation}
	\label{coneB1}
	0 \geq \frac{1}{n} \sum_{k=1}^K \sum_{i=1}^{n}   \left\{ \rho_{\tau_k}\left[y_i - \vx_i^\top \hat{\vbeta}_k\right] - \rho_{\tau_k}\left[y_i - \vx_i^\top \vbeta_{k}^*\right]\right\} + \lambda \sum_{j=1}^p  \left( ||\hat{\vbeta}^j||_2 - ||\vbeta^{*j}||_2\right).
	\end{equation}
	
	By convexity and definition of subgradients for all $k \in \{1,\ldots,K\}$, 
	\begin{equation}
	\label{coneB2}
	\rho_{\tau_k}[y_i - \vx_i^\top \hat{\vbeta}_k] - \rho_{\tau_k}\left[y_i - \vx_i^\top \vbeta_{k}^*\right] \geq - \left\{ \tau_k - I\left[y_i \leq \vx_i^\top \vbeta_{k}^*\right]\right\}\vx_i^\top[\hat{\vbeta}_k-\vbeta_{k}^*].
	\end{equation}
	Combining \eqref{coneB1} and \eqref{coneB2}
	\begin{equation}
	\label{coneB2a}
	0 \geq -\frac{1}{n} \sum_{k=1}^K  \sum_{i=1}^{n}   \left\{ \tau_k - I\left[y_i \leq \vx_i^\top \vbeta_{k}^*\right]\right\}\vx_i^\top[\hat{\vbeta}_k-\vbeta_{k}^*] + \lambda \sum_{j=1}^p  \left( ||\hat{\vbeta}^j||_2 - ||\vbeta^{*j}||_2\right).
	\end{equation}
	Using, Cauchy–Schwarz  inequality and the definitions of $\Lambda$ and $\vv_j$ it follows that 
	\begin{eqnarray*}
		\frac{1}{n} \sum_{k=1}^K  \sum_{i=1}^{n}   \left\{ \tau_k - I\left[y_i \leq \vx_i^\top \vbeta_{k}^*\right]\right\}\vx_i^\top[\hat{\vbeta}_k-\vbeta_{k}^*] 
		&=& \sum_{j=0}^p \sum_{k=1}^K \frac{1}{n} \sum_{i=1}^n x_{ij} [ \tau_k - I(\epsilon_i^k \leq 0)](\hat{\beta}_{kj}-\beta_{kj}^*) \\
		&=& \sum_{j=0}^p \vv^{j\top}(\hat{\vbeta}^j-\vbeta^{*j}) \\
		&\leq& \underset{j \in \{0,\ldots,p\}}{\max} ||\vv^j||_2 \sum_{j=0}^p ||\hat{\vbeta}^j-\vbeta^{*j}||_2 \\
		&=& \Lambda \sum_{j=0}^p ||\hat{\vbeta}^j - \vbeta^{*j}||_2.
	\end{eqnarray*}
	With $\lambda/2 \geq \Lambda$ it follows that,
	\begin{equation}
	\label{coneB3}
	\frac{1}{n} \sum_{k=1}^K  \sum_{i=1}^{n}   \left\{ \tau_k - I\left[y_i \leq \vx_i^\top \vbeta_{k}^*\right]\right\}\vx_i^\top[\hat{\vbeta}_k-\vbeta_{k}^*]
	\leq \frac{\lambda}{2} \sum_{j=0}^p ||\hat{\vbeta}^j - \vbeta^{*j}||_2. 
	\end{equation}
	Combine \eqref{coneB2a} and \eqref{coneB3} and,
	\begin{equation*}
	0 \leq \frac{\lambda}{2} \sum_{j=0}^p ||\hat{\vbeta}^j - \vbeta^{*j}||_2 + \lambda \sum_{j=1}^p  \left( ||\vbeta^{*j}||_2 - ||\hat{\vbeta}^j||_2\right).
	\end{equation*}
	Removing, $\lambda$ and adding $\frac{1}{2}\sum_{j=0}^p ||\hat{\vbeta}^j - \vbeta^{*j}||_2$ to both sides we get,
	\begin{equation}
	\label{coneB4}
	\frac{1}{2} \sum_{j=0}^p ||\hat{\vbeta}^j - \vbeta^{*j}||_2 \leq ||\hat{\vbeta}^0 - \vbeta^{*0}||_2 +  \sum_{j=1}^p  \left( ||\hat{\vbeta}^j - \vbeta^{*j}||_2 + ||\vbeta^{*j}||_2 - ||\hat{\vbeta}^j||_2\right).
	\end{equation}
	Note, $||\vbeta^{*j}||_2 \leq ||\hat{\vbeta}^j - \vbeta^{*j}||_2+||\hat{\vbeta}^{j}||_2$, and for $j \notin \mathcal{S}$ that, 
	\begin{equation*}
	||\hat{\vbeta}^j - \vbeta^{*j}||_2 + ||\vbeta^{*j}||_2 - ||\hat{\vbeta}^j||_2 = ||\hat{\vbeta}^j||_2 - ||\hat{\vbeta}^j||_{2,\vm}=0.
	\end{equation*}
	Therefore upper bound becomes
	\begin{equation*}
	||\hat{\vbeta}^0 - \vbeta^{*0}||_{2,m} + \sum_{j \in \mathcal{S}} 2||\hat{\vbeta}^j - \vbeta^{*j}||_{2,m} \leq 2 \sum_{j \in \bar{\mathcal{S}}} ||\hat{\vbeta}^j - \vbeta^{*j}||_2.
	\end{equation*}
	Updating \eqref{coneB4} with the upper bound provides
	\begin{equation*}
	\frac{1}{2} \left( \sum_{j \in \bar{\mathcal{S}}^c} ||\hat{\vbeta}^j - \vbeta^{*j}||_2 + \sum_{j \in \bar{\mathcal{S}}} ||\hat{\vbeta}^j - \vbeta^{*j}||_2 \right)  \leq 2 \sum_{j \in \bar{\mathcal{S}}} ||\hat{\vbeta}^j - \vbeta^{*j}||_2.
	\end{equation*}
	Which is equivalent to  $\sum_{j \in \bar{\mathcal{S}}^c} ||\hat{\vbeta}^j - \vbeta^{*j}||_2 \leq 3 \sum_{j \in \bar{\mathcal{S}}} ||\hat{\vbeta}^j - \vbeta^{*j}||_2$. 
\end{proof}

\subsection{Proving Theorem 3.1}
The proof relies on techniques presented in \citet{qr_group_lasso}. First, we introduce lemmas similar to those in \citet{qr_group_lasso}. 

\subsubsection{Definitions}
Following definitions are used in the proof. 
\begin{eqnarray*}
	\mathcal{C}^\delta &=& \{ \vDelta \in \Real^{K(p+1)} | \vDelta \in \mathcal{C}, ||\vDelta||_2=\delta \}, \\
	\mathcal{B}_\delta &=& \{ \vbeta \in \Real^{K(p+1)} | \vbeta-\vbeta^* \in \mathcal{C}_\delta \}, \\ 
	U_{ki}(\vbeta_k) &=&  \left\{ \rho_{\tau_k}[y_i - \vx_i^\top \vbeta_k]-\rho_{\tau_k}[y_i - \vx_i^\top \vbeta^*_k]\right\}, \\
	U_k(\vbeta_k) &=& \sum_{i=1}^{n} U_{ki}(\vbeta_k).
\end{eqnarray*}

Knight's identity, found for $\tau=1/2$ in \citet{knightId} and generalized for any $\tau$ in \citet{qrBook}, provides that
\begin{equation*}
\rho_\tau(u-v) - \rho_\tau(u) = -v\left[\tau-I(u \leq 0)\right] + \int_0^v \left[I(u \leq s) - I(u \leq 0)\right] ds.
\end{equation*}

\subsubsection{Lemmas used in proof of Theorem 3.1}
\begin{lem}
	\label{ourA2}
	If Conditions \ref{condDens}-\ref{condRest} hold, $\delta \in (0,\tilde{d})$, $\lambda/2 \geq  \Lambda$, and $||\hat{\vbeta}-\vbeta^*||_2 \geq \delta$ then 
	\begin{equation}
	\underset{\vbeta-\vbeta^* \in \mathcal{C}^\delta}{\sup} \left|\sum_{k=1}^K  U_k(\vbeta_k) - E[U_k(\vbeta_k)]\right| > \delta \left( \frac{nc_1\tilde{b}}{3} \delta -  n \lambda \sqrt{s}\right).
	\end{equation}
\end{lem}
\begin{proof}
	First, we will show that under the given conditions there exists $\vbeta-\vbeta^* \in \mathcal{C}^\delta$. If $||\hat{\vbeta}-\vbeta^*||_2 = \delta$ then by Lemma \ref{lemCone}, $\hat{\vbeta}-\vbeta^* \in \mathcal{C}^\delta$. If $||\hat{\vbeta}-\vbeta^*||_2 = \delta^* > \delta$, let $t=\frac{\delta}{\delta^*}$ and define $\tilde{\vbeta}=t\hat{\vbeta}+(1-t)\vbeta^*$.  Note, that
	\begin{equation*}
	||\tilde{\vbeta}-\vbeta^*||_2 = ||t\hat{\vbeta}+(1-t)\vbeta^*-\vbeta^*||_2 = ||t\hat{\vbeta}-t\vbeta^*||_2 = t \delta^* = \delta. 
	\end{equation*}
	Similarly, 
	\begin{eqnarray*}
		\sum_{j \in \bar{\mathcal{S}}^c} ||\tilde{\vbeta}^j-\vbeta^{*j}||_2 &=& t \sum_{j \in \bar{\mathcal{S}}^c} ||\hat{\vbeta}^j-\vbeta^{*j}||_2, \\
		\sum_{j \in \bar{\mathcal{S}}} ||\tilde{\vbeta}^j-\vbeta^{*j}||_2 &=& t \sum_{j \in \bar{\mathcal{S}}} ||\hat{\vbeta}^j-\vbeta^{*j}||_2.
	\end{eqnarray*}
	Thus, $\tilde{\vbeta}-\vbeta^* \in \mathcal{C}^\delta$. 
	%
	%
	%
	
	In addition by convexity of $L_{\lambda}$ and definitions of $\hat{\vbeta}$ and $\tilde{\vbeta}$,
	\begin{equation*}
	L_{\lambda}(\tilde{\vbeta}) \leq  t L_{\lambda}(\hat{\vbeta}) + (1-t) L_{\lambda}(\vbeta^*) \leq t L_{\lambda}(\vbeta^*) + (1-t) L_{\lambda}(\vbeta^*) = L_{\lambda}(\vbeta^*).
	\end{equation*}
	Thus, by convexity of $L_{\lambda}$ and that $\mathcal{C}$ is a cone, there exists $\vbeta$ such that for $\vbeta-\vbeta^* \in \mathcal{C}^\delta$ such that
	\begin{eqnarray*}
		0 &\geq& 
		\sum_{k=1}^K  U_k(\vbeta_k) + n\lambda \sum_{j \in \mathcal{S}^c} ||\vbeta^j||_2 - ||\vbeta^{*j}||_2 + n \lambda \sum_{j \in \mathcal{S}} ||\vbeta^j||_2 - ||\vbeta^{*j}||_2\\ 
		&=& \sum_{k=1}^K  U_k(\vbeta_k) + n\lambda \sum_{j \in \mathcal{S}^c} ||\vbeta^j||_2 + n\lambda \sum_{j \in \mathcal{S}} ||\vbeta^j||_2 - ||\vbeta^{*j}||_2 \\
		&\geq&  \sum_{k=1}^K  U_k(\vbeta_k) + n\lambda \sum_{j \in \mathcal{S}} ||\vbeta^j||_2 - ||\vbeta^{*j}||_2 \geq \sum_{k=1}^K  U_k(\vbeta_k) - n\lambda \sum_{j \in \mathcal{S}} ||\vbeta^j-\vbeta^{j*}||_2.
	\end{eqnarray*}
	Starting with the first term, given that $\vbeta-\vbeta^* \in \mathcal{C}^\delta$
	\begin{eqnarray*}
		\sum_{k=1}^K  U_k(\vbeta_k) &=& \sum_{k=1}^K E[U_k(\vbeta_k)] + U_k(\vbeta_k) - E[U_k(\vbeta_k)] \\
		&\geq& \sum_{k=1}^K E[U_k(\vbeta_k)] - \underset{\vbeta-\vbeta^* \in \mathcal{C}^\delta}{\sup} \left| \sum_{k=1}^K U_k(\vbeta_k) - E[U_k(\vbeta_k)] \right|.
	\end{eqnarray*}
	Combining the above with the immediately preceding inequalities provides, 
	\begin{equation}
	\label{intermEq}
	\underset{\vbeta-\vbeta^* \in \mathcal{C}^\delta}{\sup} \left| \sum_{k=1}^K U_k(\vbeta_k) - E[U_k(\vbeta_k)] \right| \geq \sum_{k=1}^K E[U_k(\vbeta_k)]- n\lambda \sum_{j \in \mathcal{S}} ||\vbeta^j-\vbeta^{j*}||_2
	\end{equation}

	By Knight's identity, that $E[\psi_\tau(y_i - \vx_i^\top \vbeta^*_k)]=0$, and Taylor expansion where $\tilde{s}_i \in (0,\vx_i^\top[\vbeta_k-\vbeta_k^*])$,
	\begin{eqnarray*}
		E[U_k(\vbeta_k)] &=& 
		\sum_{i=1}^{n} E\left[ \int_{0}^{\vx_i^\top(\vbeta_k-\vbeta_{k}^*)}  sf_{ki}(s)+\frac{s^2}{2}f'(\tilde{s_i}) \right] \\
		&\geq& \frac{c_1}{2} \sum_{i=1}^{n} E\{[\vx_i^\top(\vbeta_k-\vbeta_{k}^*)]^2\} - E\left[ \left|\int_{0}^{\vx_i^\top(\vbeta_k-\vbeta_{k}^*)} \frac{s^2}{2} f'(\tilde{s_i}) ds\right|\right]  \\
		&\geq& \frac{c_1}{2} \sum_{i=1}^{n} E\{[\vx_i^\top(\vbeta_k-\vbeta_{k}^*)]^2\} -\frac{c_2}{6} E\left[|\vx_i^\top (\vbeta_k-\vbeta_{k}^*)|^3\right] \\
		&=& n \left\{ \frac{c_1}{2} E\{[\vx_i^\top(\vbeta_k-\vbeta_{k}^*)]^2\} -\frac{c_2}{6} E\left[|\vx_i^\top (\vbeta_k-\vbeta_{k}^*)|^3\right] \right\} \\ 
		&=& n \Biggl\{ \frac{c_1}{2} E\{[\vx_i^\top(\vbeta_k-\vbeta_{k}^*)]^2\} \\
		&& -\frac{c_2}{6} \frac{E\left[|\vx_i^\top (\vbeta_k-\vbeta_{k}^*)|^3\right]}{E\left\{[\vx_i^\top(\vbeta_k-\vbeta_{k}^*)]^2\right]^{3/2}}E\left\{[\vx_i^\top(\vbeta_k-\vbeta_{k}^*)]^2\right\}E\left\{[\vx_i^\top(\vbeta_k-\vbeta_{k}^*)]^2\right\}^{1/2} \Biggr\} \\
		&\geq& n \left\{ \frac{c_1}{2} E\{[\vx_i^\top(\vbeta_k-\vbeta_{k}^*)]^2\} -\frac{\delta \sqrt{B}_kc_2}{6} \frac{E\left[|\vx_i^\top (\vbeta_k-\vbeta_{k}^*)|^3\right]}{E\left\{[\vx_i^\top(\vbeta_k-\vbeta_{k}^*)]^2\right\}^{3/2}}E\left\{[\vx_i^\top(\vbeta_k-\vbeta_{k}^*)]^2\right\} \right\} \\
		&\geq& n \left( \frac{c_1}{2} E\{[\vx_i^\top(\vbeta_k-\vbeta_{k}^*)]^2\} -\frac{\delta \sqrt{B}_kc_2c_1}{d_kc_2\sqrt{B}_k6}E\left\{[\vx_i^\top(\vbeta_k-\vbeta_{k}^*)]^2\right\}\right) \\
		&=& n \left( \frac{c_1}{2} E\{[\vx_i^\top(\vbeta_k-\vbeta_{k}^*)]^2\} -\frac{\delta c_1}{d_k6}E\left\{[\vx_i^\top(\vbeta_k-\vbeta_{k}^*)]^2\right\}\right) \\
		&\geq&  n \left(\frac{c_1}{2} E\{[\vx_i^\top(\vbeta_k-\vbeta_{k}^*)]^2\} -\frac{c_1}{6}E\left\{[\vx_i^\top(\vbeta_k-\vbeta_{k}^*)]^2\right\} \right) \\
		&=& \frac{nc_1}{3} E\{[\vx_i^\top(\vbeta_k-\vbeta_{k}^*)]^2\} \geq \frac{nc_1\tilde{b}}{3} ||\vbeta_k-\vbeta_{k}^*||_2^2. 
	\end{eqnarray*}
	First and second inequalities use Condition \ref{condDens}, third inequality uses Condition \ref{condX} and that $||\vbeta-\vbeta^*||_2 \leq \delta$, fourth inequality uses Condition \ref{condRest}, fifth uses that $\delta < \tilde{d} \leq d_k$. Therefore,
	\begin{equation*}
	\sum_{k=1}^K E[U_k(\vbeta_k)] \geq \sum_{k=1}^K\frac{nc_1\tilde{b}}{3} ||\vbeta_k-\vbeta_{k}^*||_2^2 = \frac{nc_1\tilde{b}}{3} \sum_{k=1}^K  ||\vbeta_k-\vbeta_{k}^*||_2^2 
	\geq \frac{nc_1\tilde{b}}{3} \delta^2. 
	\end{equation*}
	
	Recall, $s = |\bar{\mathcal{S}}|$, applying the C-S inequality of $\sum_{j \in \bar{\mathcal{S}}} ||\vbeta^j-\vbeta^{j*}||_2 \leq \sqrt{ \sum_{j \in \bar{\mathcal{S}}}1 \sum_{j \in \bar{\mathcal{S}}} ||\vbeta^j-\vbeta^{j*}||_2^2}$ provides
	\begin{equation*}
	n\lambda \sum_{j \in \bar{\mathcal{S}}} ||\vbeta^j-\vbeta^{j*}||_2 \leq  n \lambda \sqrt{s} \delta.
	\end{equation*}
	Applying the above inequality and the lower bound for $E[U_k(\vbeta_k)]$ to \eqref{intermEq} completes the proof. 
\end{proof}

\begin{lem}
	\label{ourA3}
	Assuming Conditions \ref{condDens}-\ref{condRest} hold hold then for any $\delta > 0$ and $t \geq  K\sqrt{ 8n\tilde{B}} \delta$ then
	
	\begin{equation*}
	P\left(\underset{\vbeta \in \mathcal{B}_\delta}{\sup} \left|\sum_{k=1}^K  U_k(\vbeta_k) - E[U_k(\vbeta_k)]\right| \geq t \right) \\
	\leq 16p\, \exp\left\{-\frac{[t/\delta-32\sqrt{2nK}C_x\sqrt{s+1}]^2}{128^2\left[C_x^2(s+1)n\right]}\right\}.
	\end{equation*}
\end{lem}
\begin{proof}
	Proof is similar to Lemma A3 from \citet{qr_group_lasso}. In the following we derive the upper bound for the equivalent probability of,
	$$
	P\left(\underset{\vbeta \in \mathcal{B}_\delta}{\sup} \left|K^{-1}\sum_{k=1}^K  U_k(\vbeta_k) - E[U_k(\vbeta_k)]\right| \geq tK^{-1} \right).
	$$
	Define $\tilde{U}(\vbeta) =\frac{1}{K} \sum_{k=1}^K  U_k(\vbeta_k) - E[U_k(\vbeta_k)]$. 
	Let $\vr = (r_{1},\ldots,r_{n})^\top \in \Real^{n}$ be a vector of independent Rademacher random variables independent of the data $\{\vx_i,y_i\}_{i=1}^n$. Define $P_\vr$ and $E_\vr$ as the conditional probability and expectation with respect to $\vr$ given the data. 
	
	Note, $[\rho_\tau(u-v)-\rho_\tau(u)]^2 \leq v^2$ for all $u$ and $v$. Therefore, $U_{ki}(\vbeta_k)^2 \leq (\vx_i^\top\vbeta_k)^2$. Then by Condition \ref{condX} it follows that for $\vbeta \in \mathcal{B}_\delta$ that $\frac{1}{K} \sum_{k=1}^K E[U_{ki}(\vbeta_k)^2] \leq \tilde{B} \delta^2$. Thus by Lemma 2.3.7 from \citet{vanWellner} and the assumption that $tK^{-1} \geq  \sqrt{ 8n\tilde{B}} \delta$ it follows that 
	\begin{equation*}
	P\left(\underset{\vbeta \in \mathcal{B}_\delta}{\sup} |\tilde{U}(\vbeta)| > tK^{-1} \right) \leq 4 E_\vr\left\{ P_{\vr}\left[ \underset{\vbeta \in \mathcal{B}_\delta}{\sup} \left| \sum_{i=1}^{n} r_i \frac{1}{K} \sum_{k=1}^K  U_{ki}(\vbeta_k)\right|> tK^{-1}/4\right]\right\}.
	\end{equation*}
	For any $\tilde{c}>0$, by Markov's inequality, 
	\begin{equation*}
	P_{\vr}\left[ \underset{\vbeta \in \mathcal{B}_\delta}{\sup} \left| \sum_{i=1}^{n} r_i \frac{1}{K} \sum_{k=1}^K U_{ki}(\vbeta_k)\right|> \frac{t}{4K}\right] 
	\leq \exp\left[-\frac{\tilde{c}t}{4K}\right]E_{\vr}\left\{\exp\left[ \tilde{c} \underset{\vbeta \in \mathcal{B}_\delta}{\sup} \left| \sum_{i=1}^{n} r_{i} \frac{1}{K} \sum_{k=1}^K U_{ki}(\vbeta_k)\right| \right]\right\}.
	\end{equation*}
	
	Following the logic presented in Lemma A.3 of \citet{qr_group_lasso} $|\frac{1}{K} \sum_{k=1}^K U_{ki}(\vbeta_k)|$ can be formulated as a contraction function. Thus by Theorem 4.12 from \citet{banachBook}
	\begin{equation*}
	E_{\vr}\left\{\exp\left[ \tilde{c} \underset{\vbeta \in \mathcal{B}_\delta}{\sup} \left| \sum_{i=1}^{n} r_{i} \frac{1}{K}\sum_{k=1}^K U_{ki}(\vbeta_k)\right| \right]\right\} \leq E_{\vr}\left\{\exp\left[ 2\tilde{c} \underset{\vbeta \in \mathcal{B}_\delta}{\sup} \left| \sum_{i=1}^{n} r_{i} \frac{1}{K}\sum_{k=1}^K \vx_i^\top(\vbeta_k-\vbeta_{k}^*)\right| \right]\right\}.
	\end{equation*}
	
	Define $Z_j = \sqrt{K \left(\sum_{i=1}^n x_{ij} r_i\right)^2}$ and $\tilde{\vx}_{i}^j = x_{ij}r_i \mathbf{1}_K \in \Real^K$. That is, $\tilde{\vx}_i^j$ is K replicates of $x_{ij}r_i$ and let $\tilde{\vx}^j = \sum_{i=1}^n \tilde{\vx}_{i}^j$. First, note that 
	\begin{equation*}
	||\tilde{\vx}_j||_2 = \left|\left|\sum_{i=1}^n r_ix_{ij}\mathbf{1}_K\right|\right|_2 = \sqrt{K \left(\sum_{i=1}^n x_{ij} r_i \right)^2} = \sqrt{K} \left| \sum_{i=1}^n x_{ij} r_i\right| = Z_j. 
	\end{equation*}
	
	Then, using the Cauchy-Schwarz inequality and Lemma \ref{lemCone} 
	\begin{eqnarray*}
		\left| \sum_{i=1}^{n} r_{i}K^{-1}\sum_{k=1}^K \vx_i^\top(\vbeta_k-\vbeta_{k}^*)\right| 
		&=& K^{-1} \left| \sum_{j=0}^p \tilde{\vx}^{j\top} \left( \vbeta^j - \vbeta^{j*} \right) \right| \leq K^{-1} \sum_{j=0}^p ||\tilde{\vx}^j||_2 ||\vbeta^j - \vbeta^{j*}||_2 \\
		&\leq& \underset{j \in \{0,\ldots,p\}}{\max} Z_j K^{-1} \sum_{j=0}^p ||\vbeta^j-\vbeta^{j*}||_2 \\
		&=& \underset{j \in \{0,\ldots,p\}}{\max} Z_j  K^{-1} \left[ \sum_{j \in \bar{\mathcal{S}}^c} ||\vbeta^j-\vbeta^{j*}||_2 + \sum_{j \in \bar{\mathcal{S}}} ||\vbeta^j-\vbeta^{j*}||_2 \right] \\
		&\leq& \underset{j \in \{0,\ldots,p\}}{\max} Z_j  K^{-1} \left[1+3\right] \sum_{j \in \bar{\mathcal{S}}} ||\vbeta^j-\vbeta^{j*}||_2 \\
		&\leq& \underset{j \in \{0,\ldots,p\}}{\max} Z_j  K^{-1} 4 \sqrt{s+1}\delta. 
	\end{eqnarray*}
	
	Therefore,
	\begin{eqnarray*}
		&& E_{\vr}\left\{\exp\left[ 2\tilde{c} \underset{\vbeta \in \mathcal{B}_\delta}{\sup} \left| \sum_{i=1}^{n} r_{i} K^{-1} \sum_{k=1}^K \vx_i^\top(\vbeta_k-\vbeta_{k}^*)\right| \right]\right\}\\ 
		&\leq& E_{\vr}\left\{\exp\left[ 8\tilde{c}K^{-1} \delta \sqrt{s+1} \underset{j \in \{0,\ldots,p\}}{\max} Z_j  \right] \right\} \\
		&\leq& \sum_{j=0}^p E_{\vr}\left\{\exp\left[ 8\tilde{c}K^{-1}\delta\sqrt{s+1} Z_j  \right] \right\}. 
	\end{eqnarray*}
	
	
	Then by Condition \ref{condX2} and that Rademacher random variables are independent and with a second moment of one, it follows that 
	\begin{equation*}
	E_{\vr}[Z_j^2] = E_{\vr}\left[K \left(\sum_{i=1}^n x_{ij} r_i\right)^2\right]  
	= K \sum_{i=1}^n x_{ij}^2 \leq nKC_x^2.  
	\end{equation*}
	In addition, $\underset{||\valpha||_2=1}{\sup} \sum_{i=1}^{n} \left(\valpha^\top \mathbf{1}_K x_{ij}\right)^2  \leq nC_x^2$. These two bounds, the definition of $Z_j = \left|\left|\sum_{i=1}^n r_i x_{ij} \mathbf{1}_K\right|\right|_2$, and Corollary A1 from \citet{qr_group_lasso} provide 
	\begin{equation*}
	E_{\vr}[\exp(CZ_j)] \leq 16 \exp\left(CC_x\sqrt{2nK}+4C^2nC_x^2 \right)=16\exp\left[CC_x\sqrt{n}\left(\sqrt{2K}+4CC_x\sqrt{n}\right)\right].
	\end{equation*}
	
	Using the above inequality, 
	
	\begin{eqnarray*}
		&& \sum_{j=0}^p E_{\vr}\left\{\exp\left[ 8\tilde{c}K^{-1}\delta\sqrt{s+1} Z_j  \right] \right\}\\
		&\leq& 16 p \, \exp\left[8C_x\tilde{c}K^{-1}\delta\sqrt{s+1}\sqrt{n}\left(\sqrt{2K}+32\tilde{c}C_xK^{-1}\delta\sqrt{s+1}\sqrt{n}\right)\right].
	\end{eqnarray*}
	
	Therefore, for $A=8\sqrt{2n}K^{-1/2}C_x\delta\sqrt{s+1}$ and $B=16K^{-2}\left[16C_x^2\delta^2(s+1)n\right]$, 
	
	\begin{eqnarray*}
		&& P_{\vr}\left[ \underset{\vbeta \in \mathcal{B}_\delta}{\sup} \left| \sum_{i=1}^{n} r_i K^{-1} \sum_{k=1}^K U_{ki}(\vbeta_k)\right|> t/(4K)\right]  \\
		&\leq&  16p\,\exp\left(A\tilde{c}+B\tilde{c}^2\right)\exp[-\tilde{c}t/(4K)]= 16p\,\exp\left[\tilde{c}[A-t/(4K)]+B\tilde{c}^2\right].
	\end{eqnarray*}
	
	Function is minimized at $\tilde{c}=\frac{tK^{-1}-4A}{8B}$, providing the final upper bound for the Lemma. and thus we have the following upper bound for the exponential function
	%
	%
\end{proof}

Below is Theorem 12 from \citet{grossLR}, a vector Bernstein inequality, restated here for ease of presentation. 
\begin{lem}
	\label{vecBern}
	For independent $\vx_i$ with $E[\vx_i]=0$, $V=\sum_{i=1}^n E[||\vx_i||_2^2]$, $M=\underset{i \in \{1,\ldots,n\}}{\max}||\vx_i||_2$ and $t \leq V/M$ then 
	\begin{equation}
	P\left( \left|\left|n^{-1}\sum_{i=1}^n \vx_i \right|\right|_2 \geq \sqrt{V}/n + t/n\right) \leq \exp\left(-\frac{t^2}{4V} \right). 
	\end{equation}
\end{lem}

\begin{lem}
	\label{ourA4}
	If conditions \ref{condDens}-\ref{condRest} hold and $b < nC_x\sqrt{K}$ then 
	\begin{equation*}
	P\left( \Lambda > C_x\sqrt{K}/\sqrt{n} + b/n \right) \leq 2p\exp\left(-\frac{b^2}{4nKC_x^2}\right).
	\end{equation*}
\end{lem}
\begin{proof}
	To get the desired upper bound we will apply Theorem 12 from \citet{grossLR} to establish a concentration inequity for $||\vv_j||_2$ and then apply a union bound to get the desired result for $\Lambda$. Define $\tilde{v}_{ijk}=x_{ij}[\tau_k-I(\epsilon_i^k \leq 0)]$ and $\tilde{\vv}_{ij} = (\tilde{v}_{ij1},\ldots,\tilde{v}_{ijK})^\top \in \Real^K$. Note that $\vv_j = \frac{1}{n} \sum_{i=1}^n \tilde{\vv}_{ij}$. Below are bounds needed to apply Theorem 12 from \citet{grossLR}
	\begin{eqnarray*}
		\sum_{i=1}^n E[||\tilde{\vv}_{ij}||_2^2] &=& \sum_{i=1}^n \sum_{k=1}^K E\{x_{ij}^2[\tau_k-I(\epsilon_i^k \leq 0)]^2\} \leq nKC_x^2, \\
		||\tilde{\vv}_{ij}||_2 &=& \sqrt{\sum_{k=1}^K x_{ij}^2[\tau_k-I(\epsilon_i^k \leq 0)]^2} \leq \sqrt{K}C_x.
	\end{eqnarray*}
	Applying Lemma 12 from \citet{grossLR} and assuming $b \leq n\sqrt{K}C_x$ then, 
	\begin{eqnarray*}
		P\left(||\vv_j||_2 \geq C_x\sqrt{K}/\sqrt{n} + b/n\right) \leq \exp\left(-\frac{b^2}{4nKC_x^2}\right). 
	\end{eqnarray*}
	To complete the proof, $p+1 \leq 2p$ and 
	\begin{equation*}
	P(||\Lambda||_2 \geq C_x\sqrt{K}/\sqrt{n} + b/n) \leq \sum_{j=0}^p P\left(||\vv_j||_2 \geq C_x\sqrt{K}/\sqrt{n} + b/n\right).
	\end{equation*}
	%
\end{proof}

\subsubsection{Proof of Theorem \ref{bigTheorem}}
\begin{proof}
	Define, 
	\begin{eqnarray*}
		\delta &=& \frac{3}{c_1 \tilde{b}} \left[ \max\left(\delta^*, K\sqrt{\frac{8\tilde{B}}{n}}\right)+c_2\lambda \sqrt{s}\right], \\
		t &=& \delta \left( \frac{n c_1 \tilde{b}}{3} \delta - c_2 n \lambda \sqrt{s}\right).
	\end{eqnarray*}
	
	Note these definitions satisfy $t \geq \sqrt{8\tilde{B}n}\delta$, a condition of Lemma S\ref{ourA3}, as 
	\begin{eqnarray*}
		t &=& \delta \left( \frac{n c_1 \tilde{b}}{3} \frac{3}{c_1 \tilde{b}} \left[ \max\left(\delta^*, K\sqrt{\frac{8\tilde{B}}{n}}\right)+c_2\lambda \sqrt{s}\right] - c_2 n \lambda \sqrt{s}\right) \\
		&=& \delta \left( n \left[ \max\left(\delta^*, K\sqrt{\frac{8\tilde{B}}{n}}\right)+c_2\lambda \sqrt{s}\right] - c_2 n \lambda \sqrt{s}\right) \\
		&=& n \delta \left[ \max\left(\delta^*, K\sqrt{\frac{8\tilde{B}}{n}}\right)\right] \geq K\sqrt{8\tilde{B}n}\delta.
	\end{eqnarray*}
	Using similar logic $t/\delta \geq n \delta^*$. While conditions of the Theorem guarantee that $\delta < \tilde{d}$. 
	By Lemma S\ref{ourA2} and definition of $t$,
	\begin{equation*}
	P(||\hat{\vbeta}-\vbeta^*||_2 \geq \delta \cap \lambda/2 \geq \Lambda) \leq P\left(\left|\sum_{k=1}^K  U_k(\vbeta_k) - E[U_k(\vbeta_k)]\right| > t \right).
	\end{equation*}
	
	Recall, $\delta^* = 8C^*\sqrt{\frac{s+1}{n}}\left[4D_1\sqrt{\log(p)}+\sqrt{2K}\right]$. By Lemma S\ref{ourA3}, and using that $t/\delta \geq n\delta^*$,
	\begin{eqnarray*}
		&& 
		P\left(\left|\sum_{k=1}^K  U_k(\vbeta_k) - E[U_k(\vbeta_k)]\right| > t \right) \\
		&\leq& 16p\, \exp\left\{-\frac{[t/\delta-32\sqrt{2nK}C_x\sqrt{s+1}]^2}{128^2\left[C_x^2(s+1)n\right]}\right\} \\
		&\leq& 16p\, \exp\left\{-\frac{[n\delta^*-32\sqrt{2nK}C_x\sqrt{s+1}]^2}{128^2\left[C_x^2(s+1)n\right]}\right\}=16 p^{1-D_1}. \\
	\end{eqnarray*}
	Note, the definition of $\lambda$ satisfies the conditions of Lemma S\ref{ourA4} and thus, 
	\begin{equation*}
	P(||\hat{\vbeta}-\vbeta^*||_2 \geq \delta) \leq 16p^{1-D_1}+2p^{1-D_2}.  
	\end{equation*} 
\end{proof}

\bibliographystyle{apalike}

\bibliography{acrossTau}